\documentclass{article}
    \author{Dan Betea~\thanks{\texttt{dan.betea@gmail.com}}}
    \title{Correlations for symplectic and orthogonal Schur measures}
    
    \usepackage{amsmath, amsthm, amssymb, xcolor, mathtools, bbm, mathabx}
    \usepackage{hyperref}
    \usepackage[margin=1.0in]{geometry}
    \usepackage{stmaryrd}
    \newtheorem{thm}{Theorem}
    \newtheorem{prop}[thm]{Proposition}
    \newtheorem{cor}[thm]{Corollary}

    \theoremstyle{definition}
    
    \theoremstyle{remark}
    \newtheorem{rem}[thm]{Remark}
    \numberwithin{equation}{section}
    \DeclareMathAlphabet{\mathpzc}{OT1}{pzc}{m}{it}
    \newcommand{\R}{\mathbb{R}}
    \newcommand{\Z}{\mathbb{Z}}
    \newcommand{\N}{\mathbb{N}}
    \newcommand{\C}{\mathbb{C}}
    \newcommand{\Gami}{\Gamma_{-}}
    \newcommand{\Gapl}{\Gamma_{+}}

    \newcommand{\Gapm}{\Gamma_{\pm}}
    \newcommand{\Gsp}{\Gamma_{sp+}}
    \newcommand{\Gspm}{\Gamma_{sp\pm}}
    \newcommand{\Gsm}{\Gamma_{sp-}}
    \newcommand{\Gop}{\Gamma_{o+}}
    \newcommand{\Gopm}{\Gamma_{o\pm}}
    \newcommand{\Gom}{\Gamma_{o-}}
    
    \newcommand{\Par}{\mathcal{P}}
    
    \newcommand{\F}{\mathcal{F}}

    \DeclareMathOperator{\Ad}{Ad}
    \newcommand{\bra}[1]{\langle #1 |}
    \newcommand{\ket}[1]{| #1 \rangle}
    \newcommand{\vv}{| 0 \rangle}
    \newcommand{\vcv}{\langle 0 |}

    \newcommand{\im }{\mathrm{i}}  % this is the imaginary unit
    
    \newcommand{\dx }{\mathrm{d}}
    
    \newcommand{\tr}{\mathrm{tr}}
    \newcommand{\Ap}{\mathcal{A}^{+}_{2 \to 1}}
    \newcommand{\Am}{\mathcal{A}^{-}_{2 \to 1}}
    \newcommand{\Apm}{\mathcal{A}^{\pm}_{2 \to 1}}
    \let\a\undefined
    \let\b\undefined
    \newcommand{\a}{\mathsf{a}}
    \newcommand{\b}{\mathsf{b}}
    \allowdisplaybreaks

\begin{document}

\maketitle

\abstract{We show, using either Fock space techniques or Macdonald difference operators, that certain symplectic and orthogonal analogues of Okounkov's Schur measure are determinantal with kernels given by explicit double contour integrals. We give two applications: one equates certain Toeplitz+Hankel determinants of random matrix theory with appropriate Fredholm determinants and computes Szeg\H{o} asymptotics for the former; another finds that the simplest examples of said measures exhibit discrete sine kernel asymptotics in the bulk and Airy 2 to 1 kernel---along with a certain dual---asymptotics at the edge. We believe the edge behavior to be universal.}

\section{Introduction}
\label{sec:intro}

The classical Cauchy identities for the symplectic and orthogonal characters, which pair one such function with a Schur polynomial, read as follows---see, e.g.,~\cite{kt,sun2}:
\begin{align} \label{eq:cauchy_intro}
    \sum_{\lambda} sp_{\lambda} (x_1^{\pm 1}, \dots, x_N^{\pm 1}) s_{\lambda} (y_1, \dots, y_N) &= \prod_{1 \leq i < j \leq N} (1 - y_i y_j) \prod_{1 \leq i, j \leq N} \frac{1}{(1-y_i x_j) (1-y_i x_j^{-1})}, \notag \\
    \sum_{\lambda} o_{\lambda} (x_1^{\pm 1}, \dots, x_N^{\pm 1}) s_{\lambda} (y_1, \dots, y_N) &= \prod_{1 \leq i \leq j \leq N} (1 - y_i y_j) \prod_{1 \leq i, j \leq N} \frac{1}{(1-y_i x_j) (1-y_i x_j^{-1})}, \\
    \sum_{\lambda} o_{\lambda} (x_1^{\pm 1}, \dots, x_N^{\pm 1}, 1) s_{\lambda} (y_1, \dots, y_N) &= \prod_{1 \leq i < j \leq N} (1 - y_i y_j) \prod_{1 \leq i \leq N} (1+y_i) \prod_{1 \leq i, j \leq N} \frac{1}{(1-y_i x_j) (1-y_i x_j^{-1})}. \notag
\end{align}
They can be used to construct (a priori complex) probability measures on partitions via
\begin{equation} \label{eq:measures_intro}
    m_{sp}(\lambda) \ \propto \ sp_{\lambda} (X) s_{\lambda} (Y), \qquad m_{o}(\lambda) \ \propto\ o_{\lambda} (\tilde{X}) s_{\lambda} (Y)
\end{equation}
where the alphabets $X = (x_1, x_1^{-1}, \dots, x_N, x_N^{-1})$, $Y = (y_1, \dots, y_N)$ are the parameter sets of said measures; $\tilde{X}$ stands for either $(x_1, x_1^{-1}, \dots, x_N, x_N^{-1})$ or $(x_1, x_1^{-1}, \dots, x_N, x_N^{-1}, 1)$; and the normalization constants can be read from~\eqref{eq:cauchy_intro}. 

The purpose of this note is to show that $m_{sp}$ and $m_{o}$ are determinantal measures with explicit correlation kernels, in analogy with Okounkov's Schur measure~\cite{oko}. We also give two applications: one to the study of certain Toeplitz+Hankel determinants that previously appeared in random matrix theory, and the second of an asymptotic nature where the KPZ-type kernels Airy$_{2 \to 1}$ and a certain dual appear. 

\paragraph{Background.} The Schur measure of Okounkov~\cite{oko}, along with its slightly richer cousin---the Schur process of Okounkov and Reshetikhin~\cite{or}---have, in the past twenty years, provided rich interplay between combinatorics, probability, statistical mechanics, and enumerative geometry. They are measures on partitions or sequences of partitions for which one can compute \textit{all} correlation functions as determinants~\cite{oko, or, br, a}---in particular they yield rich sources of determinantal point processes. In most cases one can use simple steepest descent asymptotic analysis to study thermodynamic limits. This has been achieved for problems such as: distributions of longest increasing subsequences in random permutations~\cite{bdj, br1, br2} and random words~\cite{tw2}; asymptotics of currents in the totally asymmetric simple exclusion process (TASEP) with both wedge and flat initial conditions~\cite{joh2, si, PF04}; last passage percolation with iid geometric or exponential weights~\cite{joh2}; last passage percolation with iid Bernoulli weights, poetically named \textit{oriented digital boiling}~\cite{gtw}; the arctic circle theorem for the Aztec diamond~\cite{bbccr}---but the approach has been implicitly used prior by Johansson~\cite{joh}; $z$-measures in the representation theory of the infinite symmetric group~\cite{oko4}; plane partitions~\cite{or}; and many others not listed here due to space constraints. We refer the reader to the review articles of Borodin--Petrov~\cite{bp2} and Borodin--Gorin~\cite{bg} for probabilistic/statistical mechanical points of view; and to the (now surely dated) review of Okounkov~\cite{oko3} for connections to enumerative geometry. 

In this note we discuss certain ``type B, C and D'' analogues of Schur measures, the ones mentioned in~\eqref{eq:measures_intro}. They are determinantal with simple correlation kernels. We give two applications, as described below.  

\paragraph{Main results.} We now outline the main results proved or reproved in this paper. We do so in words as much as possible and refer the reader to the respective sections for the precise statements. 
\newline

\textbf{Theorem A.} The measures $m_{sp/o}$ from~\eqref{eq:measures_intro} are determinantal with explicit correlation kernels $K_{sp/o}$ given by double contour integral formulae. 
\newline

This is the main result of the present paper; it is the content of Theorem~\ref{thm:corr_char} for the case of variables and of Theorem~\ref{thm:corr_lifted} for the case of arbitrary specializations. There is also a dual case for variables, Corollary~\ref{thm:corr_char_dual}.

We next consider the implications of Theorem A to certain Toeplitz+Hankel determinants which appeared in random matrix theory~\cite{joh5} and the study of symmetrized longest increasing subsequences~\cite{br1} to name but two examples---see also Remark~\ref{rem:det_history}. For that, fix a function $f(z) = \exp(R_+(z) + R_-(z)) = \sum_{k \in \Z} f_k z^k$ with $R_{\pm}(z) = \sum_{k \geq 1} \rho_k^{\pm} z^{\pm k}$ for $(\rho^{\pm}_k)_{k \geq 1}$ appropriately chosen numbers. Denote $\check{f}(z) = 1/f(-z)$ having Fourier coefficients $\check{f}_k, k \in \Z$. Consider the following Toeplitz+Hankel determinants:
\begin{equation}
    \begin{split}
        D^1_{n} = \det [f_{i-j} + f_{i+j}]_{0 \leq i,j \leq n-1}, \qquad D^2_{m} = \det [\check{f}_{i-j} - \check{f}_{i+j+2}]_{0 \leq i,j \leq m-1}, \\ 
        D^3_{n} = \det [f_{i-j} - f_{i+j+2}]_{0 \leq i,j \leq n-1}, \qquad D^4_{m} = \det [\check{f}_{i-j} + \check{f}_{i+j}]_{0 \leq i,j \leq m-1}.
        \end{split}
\end{equation}  

We can write these determinants as certain observables for the measures $m_{sp/o}$. 
\newline

\textbf{Theorem B.} We have the following Gessel-like formulae:
\begin{equation}
    \begin{split}
        \frac{1}{2} D^1_n = \sum_{\lambda: \ell(\lambda) \leq n} sp_{\lambda} (\rho^+) s_{\lambda} (\rho^-), \qquad D^2_m = \sum_{\lambda: \lambda_1 \leq m} sp_{\lambda} (\rho^+) s_{\lambda} (\rho^-), \\
        D^3_n = \sum_{\lambda: \ell(\lambda) \leq n} o_{\lambda} (\rho^+) s_{\lambda} (\rho^-), \qquad  \frac{1}{2} D^4_m = \sum_{\lambda: \lambda_1 \leq m} o_{\lambda} (\rho^+) s_{\lambda} (\rho^-).
        \end{split}
\end{equation}
\newline

This is the content of Theorem~\ref{thm:gessel} and we use $f(\rho)$ for a symmetric function specialized at $\rho$: replacing the Newton powersums $k^{-1} p_k$ by the numbers $\rho_k$.

We next have an asymptotic Szeg\H{o} result.
\newline

\textbf{Theorem C.} We have the Szeg\H{o}-like asymptotic formulae:
\begin{equation}
    \begin{split}
        \lim_{n \to \infty} \frac{1}{2} D^1_n = \lim_{m \to \infty} D^2_m &= \exp \sum_{k \geq 1} \left( k \rho^+_k \rho^-_k + \rho^-_{2k} - \frac{k(\rho^-_k)^2}{2} \right), \\
        \lim_{n \to \infty} D^3_n = \lim_{m \to \infty}  \frac{1}{2}  D^4_m &= \exp \sum_{k \geq 1} \left( k \rho^+_k \rho^-_k - \rho^-_{2k} - \frac{k(\rho^-_k)^2}{2} \right).
    \end{split}
\end{equation}
\newline

This is the content of Theorem~\ref{sec:szego} and for us it is a simple consequence of the Cauchy identities~\eqref{eq:cauchy_intro} and their generalizations to arbitrary specializations. Using different methods, the result was first proved by Johansson~\cite{joh5}, and reproved subsequently by Basor--Ehrhardt, Dehaye and Deift--Its--Krasovsky in~\cite{be1, be2, be3, be4, be5, deh, dik}. 

We also have a Borodin--Okounkov-type~\cite{bo} result of the form Toeplitz+Hankel = Fredholm---see also~\cite{tw2}.
\newline

\textbf{Theorem D.} The Toeplitz+Hankel determinants $D_m^2$ and $\frac{1}{2} D^4_m$ are proportional to Fredholm determinants of $K_{sp}$ and $K_o$ respectively. 
\newline

This is the content of Theorem~\ref{thm:bo}. We believe that, up to a reformulation, computation and somewhat different assumptions, this result first appeared in~\cite{be4}.

We finally remark that all of the aforementioned results are stated analytically but proven using essentially algebraic/combinatorial techniques. Thus a) it is entirely possible that the analytical conditions we impose for convergence can be relaxed; and b) they all have formal power series analogues which we do not discuss but see, e.g.,~\cite{bo} for a discussion of this formal aspect. 

As a different type of application, for a positive real $\theta$, we consider the measures 
\begin{equation}
    P_{sp} (\lambda) \ \propto \ sp_{\lambda} (pl_{2 \theta}) s_{\lambda} (pl_{\theta}), \qquad P_{o} (\lambda) \ \propto \ o_{\lambda} (pl_{2 \theta}) s_{\lambda} (pl_{\theta})
\end{equation}
where $pl_{\theta}$ is the Plancherel specialization sending the first powersum $p_1$ to $\theta$ and all the rest to 0. These are \textit{signed} measures whose bulk and edge asymptotic behavior (as $\theta \to \infty$) we analyze in detail. They are somewhat contrite analogues of the Poissonized Plancherel measure on partitions~\cite{oko}. Nevertheless, at least at the edge, the behavior these measures exhibit should be universal and common to all symplectic/orthogonal Schur measures having certain properties. We prove the following two theorems.
\newline

\textbf{Theorem E.} Asymptotically, the bulk behavior of $P_{sp/o}$ is governed by the discrete sine kernel.
\newline

\textbf{Theorem F.} Asymptotically, the edge behavior of $P_{sp/o}$ is governed by the Airy$_{2 \to 1}$ kernel for $P_{sp}$ and by \textit{a certain dual} for $P_o$.
\newline

These are the content of Theorem~\ref{thm:bulk_scaling} (for the bulk) and of Theorems~\ref{thm:edge_scaling} and~\ref{thm:edge_scaling_2} (for the edge).

\paragraph{Outline.} In Section~\ref{sec:prelim} we define most of what we need. The symmetric function theory is recalled in Section~\ref{sec:schur}, while Fock space rudiments are recalled in Section~\ref{sec:fermions}. In Section~\ref{sec:corr} we prove Theorem A in all its variants. For the case of alphabets (variables), we give two proofs: a self-contained one based on Fock space calculations in~\ref{sec:corr_char} and a sketch based on recent work of Ghosal~\cite{gho17} based on Macdonald difference operators in~\ref{sec:corr_another_proof}. Finally we state and prove the case of general specializations in~\ref{sec:corr_lifted}. Section~\ref{sec:gboth} covers the Gessel-like and Szeg\H{o}-like Theorems B and C in~\ref{sec:szego} and the Borodin--Okounkov-like Theorem D in~\ref{sec:bo}. Section~\ref{sec:asymptotics} is devoted to the asymptotic analysis of the bulk---Theorem E---and edge---Theorem F---of the aforementioned analogues of the Poissonized Plancherel measure. We conclude, with an outlook for future work, in Section~\ref{sec:conclusion}.

\paragraph{Notations.} Throughout, we will use the notation $\Z' := \Z + \frac{1}{2}$ whenever dealing with Fock space calculations. $\{a,b\}:=ab+ba, [a,b]:= ab-ba$ are used to denote anti-commutators and commutators respectively. $A:=(a_1, a_2, \dots), B := (b_1, b_2, \dots), Y := (y_1, y_2, \dots)$ will denote \textit{alphabets}, a pedantic name for collections of variables (or, in physicists' language, for \textit{spectral parameters}). We will need finite alphabets for \textit{BC-symmetric} Laurent polynomials of the form $X:= (x_1, x_1^{-1}, \dots, x_N, x_N^{-1})$, and we reserve $X$ for this throughout. The letter $\rho$ will always denote a specialization of the algebra of symmetric functions $\Lambda$.  

\paragraph{Acknowledgements.} The author is grateful to Estelle Basor, J\'er\'emie Bouttier, Elia Bisi, Alexei Borodin, Sasha Bufetov, L\'aszl\'o Erd\H{o}s, Patrik Ferrari, Christian Krattenthaler, Peter Nejjar, Anita Ponsaing, Eric Rains, Arun Ram, Craig Tracy, Mirjana Vuleti\'c, Paul Zinn-Justin and Nikos Zygouras for illuminating conversations and references to the literature. He is most grateful to Michael Wheeler for kick-starting this project, for partly funding it, and especially for enduring, during numerous caffeinated conversations, the author's kibitzing on blitz chess\footnote{In particular, it is in the interest of the whole mathematical community that MW's blitz chess rating does not exceed two-thousand points by year's end.} and on early explorations into Fock space. Finally remerciements should go to the Mathematics and Statistics Department at the University of Melbourne, where part of the research has been conducted during the author's antipodal visit in early two-thousand eighteen.

\section{Schur, symplectic, orthogonal and free fermions}
\label{sec:prelim}

\subsection{Schur functions, symplectic and orthogonal characters}
\label{sec:schur}

In this section we gather some generalities on Schur functions and on symplectic/orthogonal characters.

A \textit{partition} is a non-increasing sequence of non-negative integers $\lambda_1 \geq \lambda_2 \geq \dots$, called \textit{parts}, only finitely many of which are non-zero. It has \textit{size} $|\lambda| = \sum_i \lambda_i$ and its \textit{length}, denoted $\ell(\lambda)$, is the number of (non-zero) parts. For a partition $\lambda$, its \textit{conjugate}, denoted $\lambda'$, is defined by $\lambda'_i = |\{ j: \lambda_j \leq i \}|$.

We denote by $\Lambda$ the algebra of symmetric functions in a countable alphabet $Y=(y_1, y_2, \dots)$. Algebraic bases are given by the elementary and complete homogeneous symmetric functions $e_k(Y), h_k(Y), k \geq 0$, themselves defined by their generating series
\begin{equation}
    \sum_{i \geq 0} h_k(Y) t^k = \prod_{i} \frac{1}{1-t y_i} := H(Y; t), \qquad \sum_{i \geq 0} e_k(Y) t^k = \prod_{i} (1+t y_i) := E(Y; t).
\end{equation}
Note $H(Y; t) E (Y; -t) = 1$. Another useful basis is that of the Newton powersums $p_k = \sum_i y_i^k, k \geq 1$ with generating series 
\begin{equation}
    \exp \sum_{k \geq 1} \frac{p_k(Y) t^k}{k} = H(Y; t).
\end{equation}

Given alphabets $Y=(y_1, y_2, \dots), Z=(z_1, z_2, \dots)$ we use the notation
\begin{equation}
H(Y; Z) := \prod_{i,j} \frac{1}{1-y_i z_j} = \exp \sum_{k \geq 1} \frac{p_k(Y) p_k(Z)}{k}, \ E(Y; Z) := \prod_{i,j} \frac{1}{1-y_i z_j} = \exp \sum_{k \geq 1} \frac{(-1)^{k+1} p_k(Y) p_k(Z)}{k}
\end{equation}
and we also reserve the following for later purposes:
\begin{equation}
    h_{sp}(Y) := \prod_{i < j} (1-y_i y_j), \qquad h_{o}(Y) := \prod_{i \leq j} (1-y_i y_j).
\end{equation}

\textit{Schur functions} $s_{\lambda}(Y)$ form yet another linear basis of $\Lambda$ and are defined, for now and our purposes, by the Jacobi--Trudi formulae:
\begin{equation}
    \label{eq:s_jt}
    s_{\lambda}(Y) = \det [h_{\lambda_i-i+j}(Y)]_{1 \leq i, j \leq \ell(\lambda)} = \det [e_{\lambda'_i-i+j}(Y)]_{1 \leq i, j \leq \lambda_1}.
\end{equation}
Using \eqref{eq:s_jt} one could actually define \textit{skew} Schur functions as follows: for a partition $\mu \subseteq \lambda$, $s_{\lambda / \mu}(Y) = \det [h_{\lambda_i-i-(\mu_j-j)}(Y)]_{1 \leq i, j \leq \ell(\lambda)}$. When $\mu = 0$ is the \textit{empty partition}, $s_{\lambda/0} = s_{\lambda}$. For more on symmetric functions, the reader should consult Chapter I of the scriptures~\cite{mac}.

Let $X = (x_1, x_1^{-1}, \dots, x_N, x_N^{-1})$ be a finite alphabet containing variables and inverses, and let $\tilde{X}$ stand for either $X$ or $(X, 1):=(x_1, x_1^{-1}, \dots, x_N, x_N^{-1}, 1)$. Given a partition $\lambda$ of length $\leq N$, we define the \textit{$BC$-symmetric} (symmetric under permuting and inverting variables) \textit{symplectic and orthogonal characters} via the following Jacobi--Trudi formulae:
\begin{equation}
    \label{eq:sp_o_jt}
    \begin{split}
        sp_{\lambda}(X) &= \frac{1}{2} \det [h_{\lambda_i - i + j}(X) + h_{\lambda_i - i - j + 2}(X)]_{1 \leq i, j \leq \ell(\lambda)} \\
        & = \det [e_{\lambda'_i - i + j}(X) - e_{\lambda'_i - i - j}(X)]_{1 \leq i, j \leq \lambda_1}, \\
        o_{\lambda}(\tilde{X}) &= \det [h_{\lambda_i - i + j}(\tilde{X}) - h_{\lambda_i - i - j}(\tilde{X})]_{1 \leq i, j \leq \ell(\lambda)} \\
        & = \frac{1}{2} \det [e_{\lambda'_i - i + j}(\tilde{X}) + e_{\lambda'_i - i - j + 2}(\tilde{X})]_{1 \leq i, j \leq \lambda_1}.
    \end{split}
\end{equation}
They are the irreducible characters of the symplectic and orthogonal groups where the $x$'s should be viewed as eigenvalues of the corresponding conjugacy classes of matrices\footnote{In the case of the orthogonal groups, the situation depends on whether the group has even or odd rank; we use $\tilde{X}$ to stand for eigenvalues from both.}. For more on these aspects, we refer the reader to Chapters 24 and onwards of~\cite{fh}, as well as the review~\cite{sun2} and references therein.

Yet another way to define $sp$ and $o$ is in terms of the expansion in \textit{skew} Schur functions. To wit, one has
\begin{equation} \label{eq:sp_o_s}
    sp_{\lambda}(X) = \sum_{\alpha} (-1)^{|\alpha|/2} s_{\lambda / \alpha}(X), \qquad o_{\lambda}(\tilde{X}) = \sum_{\beta} (-1)^{|\beta|/2} s_{\lambda / \beta}(\tilde{X})
\end{equation}
where $\alpha$ ranges over all partitions having \textit{Frobenius coordinates} $\alpha = (a_1, a_2, \dots|a_1+1, a_2+1, \dots)$; $\beta$ ranges over all partitions having Frobenius coordinates $\beta = (b_1+1, b_2+1, \dots | b_1, b_2, \dots)$; the empty partition is in both sets; and for a partition $\lambda = (a_1, \dots | b_1, \dots)$ we recall the Frobenius coordinates are $(a_i, b_i) = (\max(\lambda_i - i, 0), \max(\lambda'_i - i, 0))$.

The important properties, for probabilistic purposes, that these functions satisfy are the Cauchy and dual Cauchy identities---see, e.g.,~\cite{kt, sun2}:
\begin{equation}
    \label{eq:sp_o_cauchy}
    \begin{split}
    \sum_{\lambda} sp_{\lambda}(X) s_{\lambda}(Y) = h_{sp} (Y) H(X; Y), \qquad \sum_{\lambda} sp_{\lambda}(X) s_{\lambda'}(Y) = h_{o} (Y) E(X; Y), \\
    \sum_{\lambda} o_{\lambda}(\tilde{X}) s_{\lambda}(Y) = h_{o} (Y) H(\tilde{X}; Y), \qquad \sum_{\lambda} o_{\lambda}(\tilde{X}) s_{\lambda'}(Y) = h_{sp} (Y) E(\tilde{X}; Y) \\
    \end{split}
\end{equation}
which are analogues of the well-known~\cite{mac} identities for Schur functions:
\begin{equation}
    \sum_{\lambda} s_{\lambda}(Y) s_{\lambda}(Z) = H(Y; Z), \qquad \sum_{\lambda} s_{\lambda}(Y) s_{\lambda'}(Z) = E(Y; Z).
\end{equation}

The symplectic and orthogonal characters admit lifts to the algebra of symmetric functions. If $Z$ is a normal alphabet in a countable number of variables (i.e., not containing any inverses), then $sp_{\lambda}(Z)$ and $o_{\lambda}(Z)$ are defined either by equation~\ref{eq:sp_o_jt} by replacing $X/\tilde{X}$ with $Z$, or by expanding $h_{sp} (Y) H(Z; Y)$ and $h_{o} (Y) H(Z; Y)$ and picking the coefficients of $s_{\lambda}(Y)$. We call these lifts the \textit{symplectic and orthogonal Schur functions}. They satisfy the natural property that if one then sets $z_{2i-1} = x_i, z_{2i} = x_i ^{-1}, 1\leq i \leq N$ and $z_j = 0, j > 2N$ then $sp_{\lambda}(Z)$ becomes the $BC$-symmetric symplectic character $sp_{\lambda}(X)$ and similarly for $o$. 

Finally we can generalize the above a bit further by forgetting the variables $X, Y, Z$ and noting these functions could be defined entirely in terms of the $e$'s, $h$'s or $p$'s. To wit, a \textit{specialization} $\rho$ is an algebra homomorphism $\rho: \Lambda \to \C$. For $f \in \Lambda$ we write $f(\rho)$ en lieu of $\rho(f)$. We note specifying $\rho$ is equivalent to specifying either of the three collections $h_{k}(\rho), e_{k}(\rho), p_{n}(\rho), k \geq 0, n \geq 1$ (with the convention that $h_0(\rho) = e_0(\rho) = 1$). The $H$ and $E$ generating series of $\rho$ are respectively
\begin{equation}
    H(\rho; t) := \sum_{k \geq 0} h_{k} (\rho) t^k = \exp \sum_{k \geq 1} \frac{p_k(\rho) t^k}{k}, \qquad E(\rho; t) := \sum_{k \geq 0} e_{k} (\rho) t^k = \exp \sum_{k \geq 1} \frac{(-1)^{k+1} p_k(\rho) t^k}{k}
\end{equation}
noting again $H(\rho; t) E(\rho; -t) = 1$. Further, if $\rho, \rho'$ are two specializations, we define
\begin{equation}
    \begin{split}
    H(\rho; \rho') &:= \exp \sum_{k \geq 1} \frac{p_k(\rho) p_k(\rho')}{k}, \\
    h_{sp}(\rho) &:= \exp \sum_{k \geq 1} \left( \frac{p_{2k}(\rho)}{2k} - \frac{p^2_{k}(\rho)}{2k} \right), \qquad h_{o}(\rho) := \exp \sum_{k \geq 1} \left( - \frac{p_{2k}(\rho)}{2k} - \frac{p^2_{k}(\rho)}{2k} \right).
    \end{split}
\end{equation} 

We remark that we think of $\rho$ in terms of the numbers $p_k(\rho), k \geq 1$ which completely determine it.

The meaning of the $s_{\lambda}(\rho), sp_{\lambda}(\rho), o_{\lambda}(\rho)$ now becomes transparent. One uses the Jacobi--Trudi identities~\eqref{eq:s_jt} and~\eqref{eq:sp_o_jt} with $\rho$ replacing $X, \tilde{X}, Y$. For $sp$ and $o$ one can also use~\eqref{eq:sp_o_s}. Finally if $\omega : \Lambda \to \Lambda$ is the involution $\omega(p_n) = (-1)^{n-1} p_n$ then one has:
\begin{equation}
    \omega (s_{\lambda}) = s_{\lambda'}, \qquad \omega (sp_{\lambda}) = o_{\lambda'}.
\end{equation}

The Cauchy identities~\eqref{eq:sp_o_cauchy} become:
\begin{equation}
    \label{eq:sp_o_cauchy_lifted}
    \begin{split}
    \sum_{\lambda} sp_{\lambda}(\rho^+) s_{\lambda}(\rho^-) &= h_{sp} (\rho^-) H(\rho^+; \rho^-) = \exp \sum_{k \geq 1} \left( \frac{p_k(\rho) p_k(\rho')}{k} + \frac{p_{2k}(\rho)}{2k} - \frac{p^2_{k}(\rho)}{2k} \right), \\
    \sum_{\lambda} o_{\lambda}(\rho^+) s_{\lambda}(\rho^-) &= h_{o} (\rho^-) H(\rho^+; \rho^-) = \exp \sum_{k \geq 1} \left( \frac{p_k(\rho) p_k(\rho')}{k} - \frac{p_{2k}(\rho)}{2k} - \frac{p^2_{k}(\rho)}{2k} \right). 
    \end{split}
\end{equation}

\subsection{Free fermions and Fock space}
\label{sec:fermions}

We recall some standard material on the theory of discrete free fermions, useful for the study of
Schur measures~\cite{oko} and processes~\cite{or}. We follow the exposition of Okounkov~\cite[Appendix A]{oko} and Baker~\cite{bak} for what is needed on symplectic and orthogonal characters and functions---see also \cite{djm} and \cite[Chapter14]{kac} for the Schur case.

The \emph{charged fermionic Fock space}, denoted $\F$, is the infinite
dimensional Hilbert space spanned by the orthonormal basis $\ket{S}$, where $S
 \subset \Z'$ with $|S \cap \Z'_+|<\infty$ and $|\Z'_-  - S| < \infty$. We call the elements of $S$ \textit{particles}. We also use the bra--ket notation throughout, denoting dual vectors by
$\bra{\cdot}$. A basis vector $\ket{S}$ could be seen as the semi-infinite wedge product
$\ket{S} = \underline{s_1} \wedge \underline{s_2} \wedge \underline{s_3} \wedge \cdots$
where $s_1 > s_2 > s_3 > \cdots$ are the \textit{particle positions}. States $S$ are in bijection with pairs $(\lambda, c)$, where $\lambda \in \Par$ is a partition and $c \in \Z$, via the following \textit{combinatorial boson--fermion} correspondence:
\begin{equation}
 (\lambda, c) \longleftrightarrow S(\lambda,c):=\{ \lambda_i - i + 1/2 + c ,\ i \geq 1\}.
\end{equation}
Using this observation, we see that $\F = \oplus_{c \in \Z} \F_c$ where $\F_c$ is spanned by partitions $\lambda$ of charge $c$. We denote by $\ket{\lambda, c}$ the vector corresponding to the pair $(\lambda, c)$, and by $\ket{\lambda}$ the vector in $\ket{\lambda, 0} \in \F_0$. Throughout this note we will only be interested in the charge 0 subspace $\F_0$. The vector $\vv$, corresponding to the empty partition (of charge 0), is called the \emph{vacuum} vector (sometimes, the \textit{Fermi sea}). 

For $k \in \Z'$, define the free fermionic (Clifford) operators 
\begin{equation}
    \psi_k \ket{S} := \underline{k} \wedge \ket{S}, \qquad 
    \psi^*_k \ket{S} := 
        \begin{cases} 
            (-1)^{i-1} s_1 \wedge \dots s_{i-1} \wedge s_{i+1} \dots, &\text{if } k=s_i \text{ for some }i,\\
            0, &\text{otherwise}.    
        \end{cases}
\end{equation}
Informally, $\psi_k$ tries to add a particle to the configuration $S$, while $\psi^*_k$ tries to remove one (with an appropriate sign in both cases), and both kill the vector if a particle is there at site $k$ (for $\psi_k$) or is absent (for $\psi^*$). Furthermore, $\psi_k$ and $\psi_k^*$ are adjoint under the inner product $\langle \lambda, c | \mu, d \rangle = \delta_{\lambda, \mu} \delta_{c, d}$ and satisfy the canonical anti-commutation relations
\begin{equation}
  \label{eq:car}
  \{ \psi_k, \psi_\ell^* \} = \delta_{k,\ell}, \qquad
  \{ \psi_k, \psi_\ell \} = \{ \psi_k^*, \psi_\ell^* \} = 0, \qquad
  k,\ell \in \Z'.
\end{equation}

The generating series of these operators, the so-called \textit{fermionic fields}, are:
\begin{equation}
 \label{eq:psigendef}
 \psi(z) := \sum_{k \in \Z'} \psi_k z^k, \qquad
 \psi^*(w) := \sum_{k \in \Z'} \psi^*_k w^{-k}.
\end{equation}
Observe that $\psi_k \vv = \psi_{-k}^* \vv = 0$ for $k<0$. 

The \textit{charge operator}, defined as $C := \sum_{k \in \Z_+} \psi_k \psi^*_k - \sum_{k \in \Z_-} \psi^*_k \psi_k$ acts diagonally on states $C \ket{\lambda, c} = c \ket{\lambda, c}$ and has $\F_0$ as kernel. The \textit{shift operator} $R$ takes $\F_{c} \to \F_{c+1}$ via $R  \ket{\lambda, c} =  \ket{\lambda, c+1}$.

We construct the 
\emph{bosonic operators} $\alpha_n$ as follows:
\begin{equation}
  \alpha_n := \sum_{k \in \Z'} \psi_{k-n} \psi_k^*, \qquad
  n = \pm 1, \pm 2, \ldots
\end{equation}
and for $n=0$, $\alpha_0 = C$. We have
$\alpha_n^*=\alpha_{-n}$, and $\alpha_n \vv=0, n>0$. The bosons satisfy the Heisenberg commutations 
\begin{equation} \label{eq:alphacommut}
  [\alpha_n,\alpha_m] = n \delta_{n,-m}
\end{equation}
and commute thusly with the fermionic fields:
\begin{equation}
    \label{eq:alphapsicommut}
[\alpha_n,\psi(z)] = z^n \psi(z), \qquad
[\alpha_n,\psi^*(w)] = - w^n \psi^*(w).
\end{equation}

Define the \emph{half-vertex operators} $\Gamma_\pm(A)$ by
\begin{equation}
  \Gamma_\pm(A) := \exp \left( \sum_{n \geq 1} \frac{p_n(A) \alpha_{\pm n}}{n} \right)
\end{equation}
where we recall $p_n(A)$ is the $n$-th Newton powersum in $A = (a_1, \dots)$. $\Gamma_-(A)$ is the adjoint of $\Gamma_+(A)$. If $A$  consists of a single variable $x$, we write $\Gamma_{\pm}(x)$. Observe that $\Gamma_{\pm}(A) = \prod_{a \in A} \Gamma_{\pm}(a)$, a consequence of 
\begin{equation}
    \label{eq:gambranch}
    \Gamma_+(A) \Gamma_+(B) = \Gamma_+(B) \Gamma_+(A).
  \end{equation}
The half-vertex operators act trivially on the vacuum and co-vacuum as follows:
\begin{equation}
  \label{eq:gamcancel}
  \Gamma_+(A) \vv = \vv, \qquad \vcv \Gamma_-(A) = \vcv.
\end{equation}
Finally, implied by~\eqref{eq:alphacommut}, by~\eqref{eq:alphapsicommut} and by the Baker--Campbell--Hausdorff formula, we have the following quasi-commutation relations:
\begin{equation}
  \label{eq:gamcomm}
  \Gamma_+(A) \Gamma_-(B) = H(A;B) \Gamma_-(B) \Gamma_+(A)
\end{equation}
and 
\begin{equation}
  \label{eq:gampsi}
  \Gamma_\pm(A) \psi(z) = H(A;z^{\pm 1}) \psi(z) \Gamma_\pm(A), \qquad
  \Gamma_\pm(A) \psi^*(w) = H(A;w^{\pm 1})^{-1} \psi^*(w) \Gamma_\pm(A).
\end{equation}

\begin{rem}
    \label{rem:Hfin}
These relations always make sense at a formal level; analytically~\eqref{eq:gamcomm} requires $|a_i b_j| < 1$; for~\eqref{eq:gampsi} we need factors of the form $1/(1-a \zeta)$ have $|a \zeta| < 1$ to be expandable in power series of $\zeta$ where $\zeta \in \{z^{\pm 1}, w^{\pm 1}\}.$
\end{rem}

The half-vertex $\Gamma$ operators commute with the charge ($C$) and shift ($R$)
 operators. They have skew Schur functions as matrix elements:
\begin{equation}
  \label{eq:schurelem}
  \bra{\lambda,c} \Gamma_+(A) \ket{\mu,d} =
  \bra{\mu,d} \Gamma_-(A) \ket{\lambda,c} =
    s_{\mu/\lambda}(A) \delta_{c,d}
\end{equation}

The fermionic operators
can be reconstructed from the half-vertex operators and the charge and
shift operators $C$ and $R$ (see e.g.\ \cite[Theorem~14.10]{kac}) via the so-called
\textit{boson--fermion correspondence}:
  \begin{equation} \label{eq:boson_fermion}
    \begin{split}  
      \psi(z) = z^{C-\frac{1}{2}} R\, \Gami(z) \Gapl^{-1} \left(z^{-1} \right),\qquad \psi^*(w) = R^{-1} w^{-C+\frac{1}{2}} \Gami^{-1}(w) \Gapl\left( w^{-1} \right).
    \end{split}
  \end{equation}

In what follows, $X=(x_1, x_1^{-1}, \dots, x_N, x_N^{-1}), Y=(y_1, \dots, y_N)$ are alphabets, taken without loss of generality to have the same number of letters. Recall that $\tilde{X}$ stands for one of $X$ or $(X, 1)$. Following Baker~\cite{bak}, let 
\begin{equation}
    \begin{split}
    \Gspm(X) &= \exp \sum_{n \geq 1} \frac{1}{n} \left( p_n(X) \alpha_{\pm n} + \frac{\alpha_{\pm 2n}}{2} - \frac{\alpha_{\pm n}^2}{2} \right), \\
    \Gopm(\tilde{X}) &= \exp \sum_{n \geq 1} \frac{1}{n} \left( p_n(\tilde{X}) \alpha_{\pm n} - \frac{\alpha_{\pm 2n}}{2} - \frac{\alpha_{\pm n}^2}{2} \right).
    \end{split}
\end{equation}
$\Gsp / \Gop$ acts trivially on the vacuum while $\Gsm / \Gom$ on the dual vacuum: e.g., $\Gsp(X) \vv = \vv$ and $\vcv \Gom(\tilde{X}) = \vcv$. They satisfy more complicated quasi-commutation relations with $\Gamma_{\pm}, \psi, \psi^*$, which we now provide along with a short proof for the benefit of the reader:
\begin{prop}
    \label{prop:sp_o_comm}
    We have the following quasi-commutation relations:
\begin{equation}
    \begin{split}
    \Gsp(X) \Gapl(Y) &= \Gapl(Y) \Gsp(X),\\
    \Gsp(X) \Gami(Y) &= H(X; Y) h_{sp}(Y) \Gami(Y) \Gapl^{-1} (Y) \Gsp(X),\\
    \Gsp(X) \psi(z) \Gsp^{-1}(X) &= H(X; z) \psi(z) \Gapl^{-1}(z),\\
    \Gsp(X) \psi^*(w) \Gsp^{-1}(X) &= (1-w^2) H(X; w)^{-1} \psi^*(w) \Gapl(w), \\
    \Gop(\tilde{X}) \Gapl(Y) &= \Gapl(Y) \Gop(\tilde{X}),\\
    \Gop(\tilde{X}) \Gami(Y) &= H(\tilde{X}; Y) h_{o}(Y) \Gami(Y) \Gapl^{-1} (Y) \Gop(\tilde{X}),\\
    \Gop(\tilde{X}) \psi(z) \Gop^{-1}(\tilde{X}) &= (1-z^2) H(\tilde{X}; z) \psi(z) \Gapl^{-1}(z),\\
    \Gop(\tilde{X}) \psi^*(w) \Gop^{-1}(\tilde{X}) &= H(\tilde{X}; w)^{-1} \psi^*(w) \Gapl(w).
    % \Gsm(X) \Gami(Y) &= \Gami(Y) \Gsm(X), \\
    % \Gapl(Y) \Gsm(X) &= H(X; Y) h_{sp}(Y) \Gsm(X) \Gami^{-1} (Y) \Gapl(Y), \\
    % \Gsm(X) \psi(z) \Gsm^{-1}(X) &= H(X; z^{-1}) \psi(z) \Gami^{-1}(z^{-1}), \\
    % \Gsm(X) \psi^*(z) \Gsm^{-1} (X) &= (1-w^{-2}) H(X; w^{-1})^{-1} \psi^*(w) \Gami(w^{-1}).
    \end{split}
\end{equation}
\end{prop}

\begin{proof}
    We will prove the first four equations. The last four follow along similar arguments. We also remark that the relations make sense formally; analytically they require similar conditions as in Remark~\ref{rem:Hfin}.
    
    The first commutation relation is immediate, as both $\Gapl$ and $\Gsp$ only involve exponentials of bosonic operators $\alpha_n$ with $n > 0$ and all such operators commute between themselves. The second could be proved from the third as follows: first, via the boson--fermion correspondence~\eqref{eq:boson_fermion} one has $\psi(z) = z^{C-1/2} R \Gami(z) \Gapl^{-1}(1/z)$; second, observing that $R$ and $C$ commute with everything else below and that again $\Gapl$ and $\Gsp$ also commute, the third equation implies
    \begin{equation}
        \Gsp(X) \Gami(z) \Gsp^{-1}(X) = H(X; z) \Gami(z) \Gapl^{-1} (z)
    \end{equation}
    and so
    \begin{equation}
        \begin{split}
        \Gsp(X) \Gami(Y) \Gsp^{-1}(X) &= \prod_{i=1}^N \Gsp(X) \Gami(y_i) \Gsp^{-1}(X)\\
        &= \prod_{i=1}^N H(X; y_i) \Gami(y_i) \Gapl^{-1}(y_i) \\
        &= H(X; Y) h_{sp}(Y) \Gami(Y) \Gapl^{-1}(Y)
        \end{split}
    \end{equation}
    where for the last equality we have commuted all $\Gapl^{-1}(y_i)$ to the right and we have collected the terms at the end.
    We now prove the third equation. We can write $\Gsp(X) = \exp(H_1 + H_2 + H_3)$ where $H_1:= \sum_{n \geq 1} p_n(X) n^{-1} \alpha_n$, $H_2:= \sum_{n \geq 1} (2n)^{-1} \alpha_{2n}$, $H_3:= \sum_{n \geq 1} (2n)^{-1} \alpha_n^2$. The $H$'s commute among themselves as they only involve positive index bosons. Then equation~\eqref{eq:alphapsicommut} implies, via Baker--Campbell--Hausdorff,
    \begin{equation}
        e^{H_1} \psi(z) e^{-H_1} = H(X; z) \psi(z), \qquad e^{H_2} \psi(z) e^{-H_2} = \frac{1}{\sqrt{1-z^2}}  \psi(z).
    \end{equation}
    Finally, for $H_3$, notice that $[\alpha_n^2, \psi(z)] = z^n \alpha_n \psi(z) + z^n \psi(z) \alpha_n = z^{2n} \psi(z) + 2 z^n \psi(z) \alpha_n$ where for the first equation we have used $[ab, c] = a[b, c] + [a, c]b$, while for the second we used~\eqref{eq:alphapsicommut} to write the $\alpha \psi$ term in reverse order. Therefore 
    \begin{equation}
        [H_3, \psi(z)] = \psi(z) (a(z) - H_+(z))
    \end{equation}
    where $a(z) := \log \sqrt{1-z^2}$, $H_+(z) := \sum_{n \geq 1} n^{-1} z^n \alpha_n = \log \Gapl(z)$. Inductively then we obtain
    \begin{equation}
        [H_3, [H_3, \dots [H_3, \psi(z)] \dots ]] = \psi(z) (a(z) - H_+(z))^n
    \end{equation}
    where above there are $n$ open commutator brackets. Then, again using Baker--Campbell--Hausdorff, we obtain
    \begin{equation}
        e^{H_3} \psi(z) e^{-H_3} = \psi(z) + [H_3, \psi(z)]+\frac{1}{2!} [H_3, [H_3, \psi(z)]] + \dots = \psi(z) e^{a(z) - H_+(z)} = \sqrt{1-z^2} \psi(z) \Gapl^{-1} (z).
    \end{equation}
    Putting it all together yields the desired result.
    The fourth equation follows along similar computations except we mention that the factor $\sqrt{1-w^2}$ survives, squared, owing to the commutation $[\alpha_n, \psi^*(w)] = -w^n \psi^*(w)$. Also the $\Gapl$ operator appearing in the end is not inverted, as $[H_3, \psi^*(w)] = \psi(w) (a(w) + H_+(w))$. We leave the details to the reader. 
\end{proof}

\begin{rem}
    Similar quasi-commutations between $\Gsm / \Gom$ and $\Gapl, \psi, \psi^*$ could also be written down. 
\end{rem}

The important fact about these \textit{twisted} half-vertex operators is that they have symplectic and orthogonal characters as matrix elements~\cite{bak}: 
\begin{equation}
    \label{eq:spoelem}
    sp_{\lambda}(X) = \vcv \Gsp(X) \ket{\lambda} = \bra{\lambda} \Gsm(X) \vv, \qquad o_{\lambda}(\tilde{X}) = \vcv \Gop(\tilde{X}) \ket{\lambda} = \bra{\lambda} \Gom(\tilde{X}) \vv.
\end{equation}

\section{Correlation functions}
\label{sec:corr}

\subsection{The symplectic and orthogonal character case}
\label{sec:corr_char}

Assume $X = (x_1, x_1^{-1}, \dots, x_N, x_N^{-1})$ and $Y=(y_1, \dots, y_N)$ satisfy the following inequalities:
\begin{equation}
    \label{eq:XY_inequality}
    1 \geq x_1 \geq x_2 \dots \geq x_N > y_1 \geq y_2 \geq \dots \geq y_N \geq 0.
\end{equation}
and recall that $\tilde{X}$ stands for either $X$ or $(X, 1)$. 

Consider the following two measures on partitions (the dependence on parameters $X, Y$ assumed implicit):
\begin{equation}
    m_{sp}(\lambda) = \frac{sp_{\lambda}(X) s_{\lambda}(Y)}{H(X;Y) h_{sp} (Y)}, \qquad m_{o}(\lambda) = \frac{o_{\lambda}(\tilde{X}) s_{\lambda}(Y)}{H(\tilde{X};Y) h_{o} (Y)}. 
\end{equation}

These are true probability measures in view of the Cauchy identities~\eqref{eq:sp_o_cauchy}, of the fact both $s_{\lambda}$ (see~\cite{mac}) and $sp_{\lambda}, o_{\lambda}$ (see~\cite{kt, sun2}) can be written as generating series of semi-standard tableaux with positive coefficients, and of the inequalities~\eqref{eq:XY_inequality}. 

Let $k_1, \dots, k_n \subset \Z'$ be particle positions. We are interested in the probability, under $m_{sp/o}$ that the particle configuration of $\lambda$ has particles at these positions. These are the so-called \textit{$n$-point correlation functions}:
\begin{equation}
    \varrho_{sp/o}(k_1, \dots, k_n) := m_{sp/o}(\{k_1, \dots, k_n \} \subset \{ \lambda_i - i + 1/2 \}).
\end{equation}

\begin{thm}
    \label{thm:corr_char}
    The measures $m_{sp}$ and $m_o$ are determinantal:
    \begin{equation}
        \varrho_{sp/o}(k_1, \dots, k_n) = \det [K_{sp/o}(k_i, k_j)]_{1 \leq i, j \leq n}
    \end{equation}
    with kernels
    \begin{equation}
        \begin{split}
        K_{sp} (a, b) &= \int_{z} \int_{w} \frac{F_{sp}(z)}{F_{sp}(w)} \frac{(1-w^2)}{(1-wz) (1-wz^{-1})} \frac{\dx z \dx w}{(2 \pi \im)^2 z^{a+3/2} w^{-b+1/2}}, \\
        K_{o} (a, b) &= \int_{z} \int_{w} \frac{F_o(z)}{F_o(w)} \frac{(1-z^2)}{(1-wz) (1-wz^{-1})} \frac{\dx z \dx w}{(2 \pi \im)^2 z^{a+3/2} w^{-b+1/2}}
        \end{split}
    \end{equation}
    where
    \begin{equation}
        F_{sp}(z) = \frac{H(X; z)}{H(Y; z) H(Y; z^{-1})}, \qquad F_{o}(z) = \frac{H(\tilde{X}; z)}{H(Y; z) H(Y; z^{-1})}
    \end{equation}
where the $z$ and $w$ contour are simple counterclockwise circles around 0 satisfying $y_1 < |w| < |z| < x_N$.
\end{thm}

\begin{proof}
    We will prove the $sp$ statement, the other being analogous. Throughout we use $F$ for $F_{sp}$ and $Z$ for the partition function $H(X; Y) h_{sp}(Y)$. First notice, using $\langle A \rangle := \vcv A \vv$ throughout, that
    \begin{equation}
        \begin{split}
        \varrho_{sp}(k_1, \dots, k_n) &= \frac{1}{Z} \left\langle \Gsp(X) \left( \prod_{i=1}^n  \psi_{k_i} \psi^*_{k_i} \right) \Gami(Y) \right\rangle \\
            &= \frac{1}{Z} \left[ \frac{z_1^{k_1} \dots z_n^{k_n}}{w_1^{k_1} \dots w_n^{k_n}} \right] \left\langle \Gsp(X) \left( \prod_{i=1}^n  \psi(z_i) \psi^* (w_i) \right) \Gami(Y) \right\rangle\\
        \end{split}
    \end{equation}
since the projector $\psi_k \psi_k^*$ picks out only the configurations satisfying $k \in \{ \lambda_i - i + 1/2 \}$. Here $[z^k] f$ denotes the coefficient of $z^k$ in $f$.  

Let us start with the one point correlator:
\begin{equation}
    \begin{split}
\frac{1}{Z} \langle \Gsp(X) \psi(z) \psi^*(w) \Gami(Y) \rangle &= \frac{1}{Z} \langle \Gsp(X) \psi(z) \Gsp^{-1}(X) \Gsp(X) \psi^*(w) \Gsp^{-1} (X) \Gsp(X) \Gami(Y) \rangle \\
&= (1-w^2) \frac{H(X; z)}{H(X; w)} \langle \psi(z) \Gapl^{-1}(z) \psi^*(w) \Gapl(w) \Gami(Y) \rangle \\
&= (1-w^2) \frac{H(X; z) H(Y; w) H(Y; w^{-1})} {H(X; w) H(Y; z) H(Y; z^{-1})} \langle \psi(z) \Gapl^{-1} (z) \psi^*(w) \rangle \\
&= \frac{F(z)}{F(w)} (1-w^2) \sqrt{\frac{w}{z}} \frac{1}{(1-wz)(1-wz^{-1})}
    \end{split}
\end{equation}
with
\begin{equation}
    \label{eq:defF}
F(z) = \frac{H(X; z)}{H(Y; z) H(Y; z^{-1})} = \prod_{i=1}^N \frac{(1-y_i z) (1-y_i z^{-1})}{(1-x_i z)(1-x^{-1}_i z)}
\end{equation}
as desired and where in the last equality we have used the boson--fermion correspondence to evaluate $\langle \psi(z) \Gapl^{-1}(z) \psi^*(w) \rangle$. 

We can now perform a similar computation for the $n$-point correlator. Using the notation $\Ad_G(A) = GAG^{-1}$ we have:
\begin{align}
    & \frac{1}{Z} \left\langle \Gsp(X) \left( \prod_{i} \psi(z_i) \psi^*(w_i)  \right) \Gami(Y) \right\rangle = \notag \\
    &= \frac{1}{Z}\left\langle \left( \prod_{i}  \Ad_{\Gsp(X)}(\psi(z_i)) \Ad_{\Gsp(X)}(\psi^*(w_i)) \right) \Gsp(X) \Gami(Y) \right\rangle \notag \\
    &= \prod_{i} (1-w_i^2) \frac{H(X; z_i)} {H(X; w_i)} \left\langle \left( \prod_{i} \psi(z_i) \Gapl^{-1}(z_i) \psi^*(w_i) \Gapl(w_i) \right) \Gami(Y) \right\rangle \notag \\
    &= \prod_{i} (1-w_i^2) \frac{H(X; z_i) H(Y; w_i)} {H(X; w_i) H(Y; z_i)} \prod_{i < j} \frac{H(w_i; z_j)}{H(z_i; z_j) H(w_i; w_j)} \prod_{i \leq j} H(z_i; w_j) \left\langle \left( \prod_{i} \psi(z_i) \psi^*(w_i) \right) \Gami(Y) \right\rangle \notag \\
    &= \prod_{i} (1-w_i^2) \frac{H(X; z_i) H(Y; w_i) H(Y; w_i^{-1})} {H(X; w_i) H(Y; z_i)H(Y; z_i^{-1})} \prod_{i < j} (1-z_i z_j) (1-w_i w_j) \prod_{i, j}\frac{1}{1-z_i w_j} \left\langle \prod_{i} \psi(z_i) \psi^*(w_i)  \right\rangle \\
    &= \prod_{i=1}^n (1-w_i^2) \frac{H(X; z_i) H(Y; w_i) H(Y; w_i^{-1})} {H(X; w_i) H(Y; z_i)H(Y; z_i^{-1})} \prod_{i < j} (1-z_i z_j) (1-w_i w_j) \prod_{i, j} \frac{1}{1-z_i w_j} \notag \\
    & \quad \times \sqrt{\prod_{i} \frac{w_i}{z_i}} \prod_{i<j} (z_i^{-1} - z_j^{-1} ) (w_i - w_j) \prod_{i,j} \frac{1}{1-w_i z_j^{-1}}  \notag \\
    &= \prod_{i} \frac{F(z_i)}{F(w_i)} (1-w_i^2) \sqrt{\frac{w_i}{z_i}} \prod_{i < j} (z_i - z_j) (w_i - w_j) (1-w_i w_j) \left( 1 - \frac{1}{z_i z_j} \right) \prod_{i, j} \frac{1}{(1-w_i z_j)(1-w_i z_j^{-1})} \notag \\
    &= \det \left[ (1-w_i^2) \sqrt{\frac{w_i}{z_j}} \frac{F(z_j)}{F(w_i)} \frac{1}{(1-w_i z_j) (1-w_i z_j^{-1})} \right]_{1 \leq i, j \leq n} \notag
\end{align}
where:
\begin{itemize}
    \item the first equality is tautologically true due to cancellations from conjugation;
    \item the second equality is true due to the action of $\Gsp$ on $\psi(z), \psi^*(w)$ using Proposition~\ref{prop:sp_o_comm};
    \item the third equality follows upon commuting all $\Gamma_+^{\pm 1}$'s to the right where they are absorbed by the vacuum;
    \item the fourth equality follows upon commuting $\Gamma_-(Y)$ to the left where it is absorbed by the vacuum;
    \item the fifth equality follows from computing $\left\langle \prod_{i} \psi(z_i) \psi^*(w_i)  \right\rangle$ using e.g., the boson--fermion correspondence~\eqref{eq:boson_fermion} and the quasi-commutation between $\Gapl$ and $\Gami$;
    \item the sixth equality uses the definition of $F$ from~\eqref{eq:defF};
    \item and the last equality follows from using the following $BC$-type Cauchy determinant:
    \begin{equation}
        \label{eq:bc_cauchy_det}
        \det_{i,j} \frac{1}{(1-w_i z_j)(1-w_i z_j^{-1})}  = \prod_{i<j} (z_i - z_j) (w_i - w_j) (1-w_i w_j) \left( 1-\frac{1}{z_i z_j} \right) \prod_{i,j} \frac{1} {(1-w_i z_j) (w_i z_j^{-1})}.
    \end{equation}
\end{itemize}

Passing to the correlation functions and using contour integrals to perform the coefficient extraction, we obtain
\begin{equation}
    \begin{split}
    \varrho_{sp}(k_1, \dots, k_n) &= \int_{z_1} \int_{w_1} \cdots \int_{z_n} \int_{w_n} \prod_{1 \leq i \leq n} \frac{w_i^{k_i-1} \dx z_i \dx w_i} {z_i^{k_i+1} (2 \pi \im)^2} \det_{1 \leq i, j \leq n} \left[  \frac{F(z_j)}{F(w_i)} \frac{  (1-w_i^2) w_i^{1/2} z_j^{-1/2} } {(1-w_i z_j) (1-w_i z_j^{-1})} \right]\\
    &= \det \left[ \int_{z} \int_{w} \frac{F(z)}{F(w)} \frac{(1-w^2)}{(1-wz) (1-wz^{-1})} \frac{\dx z \dx w}{(2 \pi \im)^2 z^{k_i+3/2} w^{-k_j+1/2}} \right]_{1 \leq i, j \leq n}
    \end{split}
\end{equation}
where the contours in the first $2n$-fold integral are simple counterclockwise circles centered around the origin satisfying
\begin{equation}
    x_N > |z_1| > |w_1| > \dots > |z_n| > |w_n| > y_1
\end{equation}
while in the second equation the $z$ and $w$ contours satisfy the conditions in the statement of the theorem. The second equation follows from the first by multilinearity of the determinant. 
\end{proof}

\subsection{Another proof of Theorem~\ref{thm:corr_char}}
\label{sec:corr_another_proof}

In this Section we give another proof of Theorem~\ref{thm:corr_char} based on Macdonald difference operators. The use of such operators to compute observables for Macdonald measures/processes has been championed by Borodin and Corwin in~\cite{bc}, but see also~\cite{bcgs} for a more systematic approach. Using this approach when one specializes to the case of Schur measures, the full determinantal correlations for Okounkov's~\cite{oko} Schur measure and the full pfaffian correlations for the (\textit{one free boundary}) pfaffian Schur measure have been computed by Aggarwal~\cite{a} and Ghosal~\cite{gho17} respectively. We follow the approach of Ghosal~\cite{gho17} since his computations, with minor modifications, will give rise to our result. We note that while Ghosal obtains pfaffian processes, we of course are dealing with determinantal ones, while using the same method with almost the same ingredients. This is by no means a coincidence and was probably observed for the first time by Stembridge~\cite[Section 7]{ste}, and utilized to great success by Baik and Rains~\cite{br1, br2}. 

\begin{proof} \textit{(difference operator proof of Theorem~\ref{thm:corr_char})}
    As usual, we prove the $sp$ case, the $o$ case following along similar lines. We only sketch the main argument and refer to Section 4 of ~\cite{gho17} for the technical details.
    
    Consider the following difference operator, acting on symmetric functions in the alphabet $Y = (y_1, \dots, y_N)$:
\begin{equation}
    D_q = \sum_{i} \prod_{j: j \ne i} \frac{q x_i - x_j}{x_i - x_j} T_i
\end{equation}
where $T_i f(y_1, \dots, y_i, \dots, y_N) := f(y_1, \dots, q y_i, \dots, y_N)$. This is the $0 < q=t < 1, r=1$ version of the $q, t$ Macdonald difference operator $D_N^r$ defined in the scriptures~\cite[Chapter VI.3]{mac}. It has Schur functions as eigenfunctions:
\begin{equation}
    D_q s_{\lambda}(Y) = \left( \sum_i q^{\lambda_i+N-i} \right) s_{\lambda} (Y).
\end{equation}

Theorem 4.1 of~\cite{gho17} gives the action of $D_q$ on functions of the form
\begin{equation}
    G(Y) = \prod_{i < j} f(y_i y_j) \prod_i g(y_i)
\end{equation}
for ``simple'' $f, g$. For Ghosal, $f(y) = (1-y)^{-1}, g(y) = H(X; y)$. For us, $g$ is unchanged (with the remark that $X$ now also contains inverses of variables), while $f(y) = 1-y$ is the reciprocal of Ghosal's choice---this later fact is the only difference between our approach and Ghosal's. Notice that with our choice, then $G(Y) = h_{sp} (Y) H(X; Y)$ is the partition function---let us hereinafter denote it $Z$, of the symplectic Schur measure under consideration. We state the result,~\cite[Theorem 4.1]{gho17}, using our notation:
\begin{equation}
    \label{eq:one_d}
    D_q Z = q^N Z \cdot \int_C \frac{1-u^2}{1-qu^2} \frac{F(qu)}{F(u)} \frac{\dx u}{2 \pi \im u}
\end{equation}
with (recall!) $F(u):= H(X; u) H(Y; u)^{-1} H(Y; u^{-1})^{-1}$ and with $C$ a disjoint union of contours $C_i$ each circling $y_i$ and nothing else and each having the property $q C_i$ lies outside $C_i$. Let us further use the notation, for a function $a$ defined on partitions, 
\begin{equation}
    \langle a \rangle := \mathbb{E}_{m_{sp}} (a) = \frac{1}{Z}  \sum_{\lambda} a(\lambda) sp_{\lambda} (X) s_{\lambda} (Y).
\end{equation}
Then, given equation~\ref{eq:one_d}, we have computed $Z^{-1} D_q Z = \left\langle \sum_{i} q^{\lambda_i + N - i} \right\rangle$ as a single contour integral. 

Furthermore, by picking radii $r_1 > r_2 > \dots > r_n > 0$ small enough and $q_1, \dots, q_n$ smaller than 1, we can compute---mutatis-mutandis as Ghosal has done in~\cite[Theorem 4.4]{gho17}, the following:
\begin{equation}
    \label{eq:many_d}
    \begin{split}
    \left( \prod_{j=1}^n D_{q_j} \right) Z = Z \int_{C^1} \dots \int_{C^n} \prod_{1 \leq j < l \leq n} \frac{(q_j u_j - q_l u_l) (u_j - u_l) (1 - q_j q_l u_j u_l) (1-u_j u_l)}{(u_j - q_l u_l) (q_j u_j - u_l) (1 - q_l u_j u_l) (1 - q_j u_j u_l)}  \\
    \times \prod_{j=1}^n \frac{F(q_j u_j)}{F(u_j)} \frac{q_j^N (1-u_j^2)}{(1-q_j u_j^2) (q_j u_j - u_j)} \frac{\dx u_j}{2 \pi \im}
    \end{split}
\end{equation}
where $C^j$ are disjoint unions of contours $C^j_i$, each $C^j_i$ a small circle of radius $r_j$ around $y_i$ such that $q C^j_i$ lies outside $C^j_i$ for all $i, j$. The difference with~\cite[Theorem 4.4]{gho17} is simple: every factor there of the form $1-a$ that appears in the numerator (of the integrand) must be flipped and put in the denominator in our case and vice versa. These factors come from the function $f$, which in our case is the reciprocal of Ghosal's. Note that with a bit of massaging the integrand above becomes
\begin{equation}
    \begin{split}
    &\prod_{1 \leq j<l \leq n} (u_j - u_l) (q_j u_j - q_l u_l) \left(1 - \frac{1}{q_j q_l u_j u_l}\right) (1-u_j u_l) \\
    &\prod_{1 \leq j, l \leq n} \frac{1}{(1-u_l q_j u_j) (1-u_l (q_j u_j)^{-1})} \prod_{j=1}^n \frac{q_j^N (1-u_j^2) F(q_j u_j)}{F(u_j)} \frac{\dx u_j}{2 \pi \im u_j} = \\
    &\det \left[ \frac{F(q_j u_j)}{F(u_l)} \frac{q_l^N (1-u_l^2) }{(1-u_l q_j u_j) (1-u_l (q_j u_j)^{-1})} \frac{\dx u_j}{2 \pi \im u_j} \right]_{1 \leq j, l \leq n} 
    \end{split}
\end{equation}
where for the last equation we have used again the $BC$-Cauchy determinant~\eqref{eq:bc_cauchy_det}. Further using the multilinearity of the determinant and putting it all together yields
\begin{equation}
    \begin{split}
    \label{eq:laplace_t}
    \left\langle \prod_{j=1}^n \sum_{i} q_j^{\lambda_i + N - i} \right\rangle &= Z^{-1} \left( \prod_{j=1}^n D_{q_j} \right) Z  \\
    &= \det \left[ \int_C \frac{F(q_j u_j)}{F(u_l)} \frac{q_l^N (1-u_l^2) }{(1-u_l q_j u_j) (1-u_l (q_j u_j)^{-1})} \frac{\dx u_l}{2 \pi \im u_l} \right]_{1 \leq j, l \leq n} 
    \end{split}
\end{equation}
where $C$ is as before and where we begin to recognize the form of the kernel we're after. To now finish the proof, if $k_1, \dots, k_n \in \Z$ are particle positions\footnote{Note the absence of the $1/2$ shift since we are not in Fock space anymore.}, the $n$-point correlation function is the coefficient of $\prod_{j=1}^n q_j^{-k_j-N}$ in the above, and this can be achieved by $n$ contour integrations on a contour $C'$ around zero of radius $<1$:
\begin{equation}
    \label{eq:prelim_ker}
    m_{sp}(\{k_1, \dots, k_n\} \subset \{ \lambda_i - i\}) = \int_{C'} \dots \int_{C'} \left\langle \prod_{j=1}^n \sum_{i} q_j^{\lambda_i + N - i} \right\rangle \prod_{j=1}^n \frac{\dx q_j}{2 \pi \im q_j^{N+k_j+1}}.
\end{equation}
We note that the factors $q_j^N$ cancel. With the substitution $z_j:=q_j u_j, w_j:=u_j$ and noting that $q_j^{-k_j} = w_j^{k_j} z_j^{-k_j}$, we finally obtain, up to the unimportant $1/2$ shift, the stated integrand of the kernel with one difference: the contours are not the stated ones. However the contours we have here can be expanded into the correct ones, without changing the value of the integral by a careful analysis---see~\cite{a, gho17} for the details, as they are somewhat technical. This finishes the proof. 
\end{proof}

\begin{rem}
    We compare and contrast the two proofs as follows. The difference operator proof of this section is conceptually simple yet dealing with the contours of integration is rather involved, and indeed we have referenced the original works for the details. The fermionic proof requires the heavy machinery of Fock space, yet the contours fall out naturally from the quasi-commutations~\eqref{eq:gamcomm} and~\eqref{eq:gampsi}. However, the difference operator proof yields observables in non-determinantal cases as well, and this was indeed the original idea of Borodin and Corwin~\cite{bc}. In particular, the equation~\eqref{eq:laplace_t} and its ilk are of particular interest in determinantal and non-determinantal cases alike as they are Laplace transforms of the particle positions and through them one can obtain \textit{edge} asymptotics even in the absence of the full $n$-point correlation function. The Fock space proof can be used for orthogonal purposes: investigating the $\tau$-function structure of correlation functions~\cite{oko, az} and its relations to enumerative geometry, matrix models, and the like. This has indeed been hinted at already by Baker~\cite{bak}.
\end{rem}

\subsection{The lifted case}
\label{sec:corr_lifted}

Fix two specializations $\rho^+, \rho^-$ and consider the following (possibly complex) probability measures on partitions
\begin{equation}
    m_{sp}(\lambda) = \frac{sp_{\lambda} (\rho^+) s_{\lambda} (\rho^-)}{Z_{sp}}, \qquad m_{o}(\lambda) = \frac{o_{\lambda} (\rho^+) s_{\lambda} (\rho^-)}{Z_{o}} 
\end{equation}
where 
\begin{equation} \label{eq:z_sp_o}
    \begin{split}
    Z_{sp} &= \exp \sum_{k \geq 1} \left( \frac{p_k (\rho^+) p_k (\rho^-)}{k} + \frac{p_{2k} (\rho^-)}{2k} - \frac{p_k^2(\rho^-)}{2k} \right), \\
    Z_{o} &= \exp \sum_{k \geq 1} \left( \frac{p_k (\rho^+) p_k (\rho^-)}{k} - \frac{p_{2k} (\rho^-)}{2k} - \frac{p_k^2(\rho^-)}{2k} \right).
    \end{split}
\end{equation}

\begin{thm}
    \label{thm:corr_lifted}
    Assume that the specializations $\rho^{\pm}$ are such that $k^{-1} p_k(\rho^{\pm}) = O(r_{\pm}^k)$ for some constants $r_+, r_- > 0$ with $\min(1, 1/r_+) > r_- > 0$. Then the measures $m_{sp}$ and $m_o$ are determinantal:
    \begin{equation}
        m_{sp/o} (\{k_1, \dots, k_n\} \subset \{ \lambda_i - i + 1/2\}) = \det [K_{sp/o} (k_i, k_j)]_{1 \leq i,j \leq n}
    \end{equation}
    with correlation kernels given by 
\begin{equation}
    \label{eq:kernel_lifted}
    \begin{split}
    K_{sp}(a, b) &= \int_{z} \int_{w} \frac{F(z)}{F(w)} \frac{(1-w^2)}{(1-wz) (1-wz^{-1})} \frac{\dx z \dx w}{(2 \pi \im)^2 z^{a+3/2} w^{-b+1/2}}, \\
    K_{o}(a, b) &= \int_{z} \int_{w} \frac{F(z)}{F(w)} \frac{(1-z^2)}{(1-wz) (1-wz^{-1})} \frac{\dx z \dx w}{(2 \pi \im)^2 z^{a+3/2} w^{-b+1/2}}
    \end{split}
\end{equation}
where 
\begin{equation}
    F(z) = \frac{H(\rho^+; z)}{H(\rho^-; z) H(\rho^-; z^{-1})}
\end{equation}
and the contours are simple closed counterclockwise curves around 0 such that 
\begin{equation}
    \min(1/|w|, 1/r_+, 1/r_-) > |z| > |w| > r_- > 0.
\end{equation} 
\end{thm}

\begin{rem}
    \label{rem:general_contours}
    The conditions on the contours can be relaxed. We need that $|wz|, |wz^{-1}|<1$ to be able to expand  the factors $1/(1-wz), 1/(1-wz^{-1})$ in power series of $wz, wz^{-1}$ respectively; we also need conditions linking $\rho^{\pm}$ and $z,w$ ensuring that $H(\rho^+; z), H(\rho^+; w), H(\rho^-; z^{\pm 1}), H(\rho^-; w^{\pm 1})$ are finite and expandable in power series of $z$, $w$, $z^{\pm 1}$ and $w^{\pm 1}$ respectively. One fairly general set of conditions that works is, following~\cite{bbccr}, that the $z$ contour should encircle 0 and all the negative poles of $F$ but none of the positive ones, while the $w$ contour should encircle 0 and all the positive zeroes of $F$ but none of the negative ones. In addition, of course we need $|w|^{-1} > |z| > |w|$. Finally, if the partition function ($Z_{sp}$ or $Z_o$ in our case) is finite, such contours will always exist. See, e.g.,~\cite{bbccr} for more details.
\end{rem}

\begin{proof}
The Fock space proof of the case of alphabets, Theorem~\ref{thm:corr_char}, goes through with the following modifications. First one defines
\begin{equation}
    \Gapm(\rho) = \exp \sum_{k \geq 1} \frac{p_k(\rho)}{k} \alpha_{\pm k}, \qquad \Gamma_{sp/o+}(\rho) = \exp \sum_{k \geq 1} \left( \frac{p_k(\rho)}{k} \alpha_{k} \pm \frac{\alpha_{2k}}{2k} - \frac{\alpha^2_{k}}{2k}\right)
\end{equation}
where in the second equation we take $+$ for $sp$ and $-$ for $o$. The quasi-commutations of equations~\eqref{eq:gamcomm}, \ref{eq:gampsi}, and Proposition~\ref{prop:sp_o_comm} still hold if one replaces the alphabets by arbitrary specializations and the proofs are similar. With the obvious modifications in equations~\eqref{eq:schurelem} and~\eqref{eq:spoelem}, we have that $s_{\lambda} (\rho)$, $sp_{\lambda} (\rho)$, $o_{\lambda} (\rho)$ are matrix elements of $\Gami(\rho), \Gsp(\rho), \Gop(\rho)$ respectively. Then the same steps as in the proof of Theorem~\ref{thm:corr_char} can be taken for the present purposes, provided that all the $H$ and $h_{sp/o}$ factors appearing are finite. This fact manifests itself in our choice of contours.
\end{proof}

Getting back to the case of specializations into alphabets, $\rho^+ \in \{X, \tilde{X} \}, \rho^- = Y$, we have a dual version of Theorem~\ref{thm:corr_char}. Let $m'_{sp/o}$ be the following measures:
\begin{equation}
    m'_{sp}(\lambda) = \frac{sp_{\lambda}(X) s_{\lambda'}(Y)}{H(X;Y) h_{o} (Y)}, \qquad m'_{o}(\lambda) = \frac{o_{\lambda} (\tilde{X}) s_{\lambda'}(Y)}{H(\tilde{X};Y) h_{sp} (Y)}.
\end{equation}
Again these are true probability measures provided for all $i$ we have $x_i, y_i > 0, 1 > y_i$ due to the tableaux definition of $s, sp, o$. Let $k_1, \dots, k_n \subset \Z'$ be particle positions.

\begin{cor}
    \label{thm:corr_char_dual}
    The measures $m'_{sp}$ and $m'_o$ are determinantal:
    \begin{equation}
        m'_{sp/o}(\{k_1, \dots, k_n \} \subset \{ \lambda_i - i + 1/2 \}) = \det [K_{sp/o}(k_i, k_j)]_{1 \leq i, j \leq n}
    \end{equation}
    with kernels $K_{sp/o}$ like those in Theorem~\ref{thm:corr_char} with
    \begin{equation}
        F_{sp}(z) = \frac{E(Y; z) E(Y; z^{-1})}{E(X; z)}, \qquad F_{o}(z) = \frac{E(Y; z) E(Y; z^{-1})}{E(\tilde{X}; z)};
    \end{equation}
with the $w$ contour a counterclockwise circle around zero containing none of the points $-y_i^{\pm 1}$; the $z$ contour a counterclockwise circle around 0 containing all of the points $-x_i^{\pm 1}$; and with $|w| < |z| < |w^{-1}|$. 
\end{cor}

\begin{rem}
    The contours above clearly exist if, e.g., inequalities~\eqref{eq:XY_inequality} are satisfied.
\end{rem}

\section{Gessel, Szeg\H{o}, Borodin--Okounkov for Toeplitz+Hankel}
\label{sec:gboth}

\subsection{Gessel and Szeg\H{o} formulae}
\label{sec:szego}

Let $f(z)$ be an integrable function on the unit circle having Wiener--Hopf factorization of the form
\begin{equation}
    f(z) = f_+(z) f_-(z), \qquad f_{\pm} = e^{R_{\pm}}, \qquad R_{\pm} (z) = \sum_{k \geq 1} \rho^{\pm}_k z^{\pm k}
\end{equation}
and let $f_k$ be its Fourier coefficients: $f(z) = \sum_{k \in \Z} f_k z^k.$ Unless otherwise stated, we shall assume throughout that 
\begin{equation}
    \label{eq:f_conditions}
    (\log f)_{\pm}:=R_{\pm} \text{\ are\ bounded\ and\ } \sum_{k \in \Z} |k| |(\log f)_k|^2 = \sum_{k \geq 1} k (|\rho^+_k|^2 + |\rho^-_k|^2) < \infty. 
\end{equation}
We note that $f$ can be thought of as a pair of specializations $(\rho^+, \rho^-)$ of the algebra of symmetric functions via the formula
\begin{equation}
    \frac{p_k(\rho^{\pm})}{k} = \rho^{\pm}_k.
\end{equation}
Denote $\check{f}(z) := \frac{1}{f(-z)}$ and recall the constants from~\eqref{eq:z_sp_o}, written in terms of the Fourier coefficients of $\log f$ this time:
\begin{equation}
    \label{eq:z_sp_o_2}
    Z_{sp} = \exp \sum_{k \geq 1} \left( k \rho^+_k \rho^-_k + \rho^-_{2k} - \frac{k(\rho^-_k)^2}{2} \right), \qquad Z_{o} = \exp \sum_{k \geq 1} \left( k \rho^+_k \rho^-_k - \rho^-_{2k} - \frac{k(\rho^-_k)^2}{2} \right).
\end{equation}
They are finite due to assumptions~\eqref{eq:f_conditions}.

Consider the following four Toeplitz+Hankel determinants (note the indexing starts from 0):
\begin{equation}
    \begin{split}
    D^1_{n} &= \det [f_{i-j} + f_{i+j}]_{0 \leq i,j \leq n-1},\\
    D^2_{m} &= \det [\check{f}_{i-j} - \check{f}_{i+j+2}]_{0 \leq i,j \leq m-1},\\
    D^3_{n} &= \det [f_{i-j} - f_{i+j+2}]_{0 \leq i,j \leq n-1},\\
    D^4_{m} &= \det [\check{f}_{i-j} + \check{f}_{i+j}]_{0 \leq i,j \leq m-1}.
    \end{split}
\end{equation}

\begin{rem} \label{rem:det_history}
Such determinants have appeared before: in the enumeration of nonintersecting lattice paths and plane partitions~\cite{ste}; in the context of random matrices from the classical groups~\cite{joh5}; as distributions for symmetrized longest increasing subsequences in the work of Baik--Rains~\cite{br1}; and have been studied extensively from a Szeg\H{o} asymptotics point of view in the works of Basor--Ehrhardt~\cite{be1, be2, be3, be4, be5}, Forrester--Frankel~\cite{ff} and Deift--Its--Krasovsky~\cite{dik}. Note the two constants $Z_{sp}, Z_{o}$ are the same as $E F_{IV}$ and $E F_{III}$ respectively in~\cite{be4}.
\end{rem}  

\begin{thm}
    \label{thm:gessel}
We have the following Gessel-like formulae:
\begin{equation}
    \begin{split}
    \frac{1}{2} D^1_n &= \sum_{\lambda: \ell(\lambda) \leq n} sp_{\lambda} (\rho^+) s_{\lambda} (\rho^-), \\
    D^2_m &= \sum_{\lambda: \lambda_1 \leq m} sp_{\lambda} (\rho^+) s_{\lambda} (\rho^-), \\
    D^3_n &= \sum_{\lambda: \ell(\lambda) \leq n} o_{\lambda} (\rho^+) s_{\lambda} (\rho^-), \\
    \frac{1}{2} D^4_m &= \sum_{\lambda: \lambda_1 \leq m} o_{\lambda} (\rho^+) s_{\lambda} (\rho^-).
    \end{split}
\end{equation}
\end{thm}

\begin{proof}
    We use the Jacobi--Trudi identities for Schur functions
    $s_{\lambda}(\rho) = \det [h_{\lambda_i-i+j}(\rho)] = \det [e_{\lambda'_i-i+j}(\rho)]$
    where if $\ell(\lambda) \leq n$ the first determinant can be taken to be $n \times n$ and if $\ell(\lambda') = \lambda_1 \leq m$ the second determinant can be taken to be $m \times m$. To prove the four identities we also recall corresponding Jacobi--Trudi formulae for $sp$ and $o$ functions: 
    \begin{equation}
        \begin{split}
            sp_{\lambda}(\rho) &= \frac{1}{2} \det [h_{\lambda_i - i + j}(\rho) + h_{\lambda_i - i - j + 2}(\rho)] \\
            & = \det [e_{\lambda'_i - i + j}(\rho) - e_{\lambda'_i - i - j}(\rho)], \\
            o_{\lambda}(\rho) &= \det [h_{\lambda_i - i + j}(\rho) - h_{\lambda_i - i - j}(\rho)] \\
            & = \frac{1}{2} \det [e_{\lambda'_i - i + j}(\rho) + e_{\lambda'_i - i - j + 2}(\rho)]
        \end{split}
    \end{equation}
    where again the $h$- (respectively $e$-) determinants are $n\times n$ (respectively $m\times m$) if $\ell(\lambda) \leq n$ (respectively $\lambda_1 \leq m$). We will prove the first identity, and make remarks about how to prove the second. The remaining two are completely analogous and we omit the details. The idea, following Gessel~\cite{ges} and Tracy--Widom~\cite{tw2}, is to use the Cauchy--Binet identity. That is, let $(A_{i,j})_{1 \leq i \leq n; 1 \leq j < \infty}$ be the $n \times \infty$ matrix $A_{i,j} = h_{j-i} (\rho^+)$, and $(B_{i,j})_{1 \leq i < \infty; 1 \leq j \leq n}$ be the $\infty \times n$ matrix $B_{i,j} = h_{i-j}(\rho^-) + h_{i+j-2n}(\rho^+)$. Then 
    \begin{equation}
        \det AB = \sum_{L} \det A_L \cdot \det B_L
    \end{equation}
    where the sum is over all sets $L = \{l_1 < l_2 < \dots < l_n \}$ with $1 \leq l_1$ and $A_L, B_L$ are the $n \times n$ submatrices of $A$ and $B$ with columns (respectively rows) chosen from the set $L$. On the left hand side, if one writes $l_{i} - i = \lambda_{n+1-i}$ and reverses rows and columns in both $A$ and $B$ (i.e., $i, j \mapsto n+1-i, n+1-j$), one obtains, via the $h$-Jacobi--Trudi formulae for Schur and symplectic Schur functions, $2 \sum_{\lambda: \ell(\lambda) \leq n} sp_{\lambda} (\rho^+) s_{\lambda} (\rho^-)$. On the right, the matrix entry is 
    \begin{equation}
        \begin{split}
        (AB)_{i,j} &= \sum_{k \in \Z} h_{k-i} (\rho^+) h_{k-j} (\rho^-) + h_{k-i} (\rho^+) h_{k+j-2n} (\rho^-) \\
        &= \sum_{k \in \Z} h_{k-i+j} (\rho^+) h_{k} (\rho^-) + h_{k-i-j+2n} (\rho^+) h_{k} (\rho^-), \qquad 1 \leq i, j \leq n\\
        &= \sum_{k \in \Z} h_{k+i-j} (\rho^+) h_{k} (\rho^-) + h_{k+i+j} (\rho^+) h_{k} (\rho^-), \qquad 0 \leq i, j \leq n-1
        \end{split}
    \end{equation}
    where: in the first equality we used $h_k = 0$ for $k < 0$ to sum over $\Z$; in the second equality we used the substitutions $k-j \mapsto k$ and $k+j-2n \mapsto k$ for the first and second sums respectively; and in the third equality we switched rows and columns and indexed from zero using $(n-i, n-j) \mapsto (i,j)$. To obtain $D^1_n$ we are left to check the symbol (the function $f$) is the correct one. Calculemus:
    \begin{equation}
        \sum_s \sum_k h_{k+s} (\rho^+) h_k (\rho^-) z^s = H(\rho^+; z) H(\rho^{-}; z^{-1}) = \exp \sum_{k \geq 1} (\rho^- z^{-k} + \rho^+_k z^k) = f(z)
    \end{equation}
    as desired. 
    
    For the other determinants, things proceed similarly, except the matrix $B$ changes in every case. For $D^2_m$ we use the $e$-Jacobi--Trudi formulae for Schur and symplectic Schur functions, and the matrix $B$ is given by $B_{i,j} = e_{i-j}(\rho^-) - e_{i+j-2n-2}(\rho^-)$. At the end of an analogous computation, the symbol can be computed as follows:
    \begin{equation}
        \sum_s \sum_k e_{k+s} (\rho^+) e_k (\rho^-) z^s = E(\rho^+; z) E(\rho^{-}; z^{-1}) = \frac{1}{H(\rho^+; -z) H(\rho^{-}; -z^{-1})} = \frac{1}{f(-z)} = \check{f}(z)
    \end{equation}
    as desired. For the third and fourth identities we list $B$ in both cases: $B_{i,j} = h_{i-j}(\rho^-) - h_{i+j-2n-2}(\rho^-)$ and $B_{i,j} = e_{i-j}(\rho^-) + e_{i+j-2n}(\rho^-)$ respectively and mention $A$ is the Jacobi--Trudi $h$-matrix (for Schur functions) in the third case, and the $e$-matrix in the fourth case.
\end{proof}

\begin{thm}
    \label{thm:szego}
We have the following Szeg\H{o}-type asymptotics:
\begin{equation}
    \begin{split}
    \lim_{n \to \infty} \frac{1}{2} D^1_n = \lim_{m \to \infty} D^2_m &= Z_{sp},\\
    \lim_{n \to \infty} D^3_n = \lim_{m \to \infty}  \frac{1}{2}  D^4_m &= Z_{o}.
    \end{split}
\end{equation}
\end{thm} 

\begin{proof}
    Given Theorem \ref{thm:gessel} and as already observed by Tracy--Widom and Borodin--Okounkov~\cite{tw2, bo} in the context of Toeplitz determinants and the classical strong Szeg\H{o} limit theorem, the limits are consequences of the Cauchy identities
    \begin{equation}
        \sum_{\lambda} sp_{\lambda} (\rho^+) s_{\lambda} (\rho^-) = Z_{sp}, \qquad \sum_{\lambda} o_{\lambda} (\rho^+) s_{\lambda} (\rho^-) = Z_{o}. \qedhere
    \end{equation}
\end{proof}

\subsection{Borodin--Okounkov formulae}
\label{sec:bo}

In~\cite{bo} Borodin and Okounkov have proved that, roughly speaking, Toeplitz determinants are Fredholm determinants of a kernel that has an explicit double contour integral form. In this section we do the same for (two of) the Toeplitz+Hankel determinants introduced in Section~\ref{sec:szego}. The reader should compare this with the Fredholm formulae given for these determinants in~\cite{be4}.

\begin{thm} \label{thm:bo}
    We have:
    \begin{equation}
        \begin{split}
        D_m^2 &= Z_{sp} \cdot \det (1-K_{sp})_{\ell^2 \{ m+\frac{1}{2}, m+\frac{3}{2}, \dots \} },\\
        \frac{1}{2} D_m^4 &= Z_{o} \cdot \det (1-K_o)_{\ell^2 \{ m+\frac{1}{2}, m+\frac{3}{2}, \dots \} }
        \end{split}
    \end{equation}
    where $Z_{sp/o}$ are as in equation~\eqref{eq:z_sp_o_2} and $K_{sp/o}$ as in~\eqref{eq:kernel_lifted} and where we assume the same conditions on $\rho^{\pm}$ as in Theorem~\ref{thm:corr_lifted}.
\end{thm}

\begin{proof}
We prove the $sp$ statement, the other being analogous. As was shown in Section~\ref{sec:szego}, $D_m^2/Z_{sp}$ is the probability, under $m_{sp}$, that the partition $\lambda$ has all its parts $\leq m$ or that the associated process has a last particle at $m - \frac{1}{2}$. The latter is the stated Fredholm determinant since the associated process is determinantal, as per Theorem~\ref{thm:corr_lifted}, provided that the Fredholm determinant on the right is well-defined.

To show this and to somewhat simplify notation, we identify $\{ m+\frac{1}{2}, m+\frac{3}{2}, \dots \}$ with $\N$ in the obvious way. Assume the contour integrals in $K$ (either $K_{sp}$ or $K_o$) are over circles of the form $|w| = r_-, |z| = r_+$ with $r_- < r_+ < 1/r_-$. We have $r_- < 1$. Assume first we can take $r_+ > 1$. Then as $K(a,b) = \int \int \dx z \dx w f(z, w) w^b z^{-a}$ with $f$ continuous on the contours, we have, for some constant $C$, that
\begin{equation}
    \sum_{a, b \geq 0} |K(a, b)| \leq C \cdot \sum_{a, b \geq 0} r_-^b r_+^{-a} < \infty
\end{equation}
showing $K$ is trace class so $\det (1-K)_{\ell^2 \N}$ is well-defined.

For $r_+ < 1$ the situation is slightly more complicated. Consider the conjugated kernel $K_c (a, b) = r_+^a r_+^{-b} K(a, b)$. Then $\det [K(k_i, k_j)]_{i,j} = \det [K_c(k_i, k_j)]_{i,j}$ and $\det (1-K_c)_{\ell^2 \N}$ is finite, as can be seen from a Hadamard-type argument. First, as before, 
\begin{equation}
    |K_c (a, b)| \leq C \cdot \left(\frac{r_-}{r_+}\right)^b
\end{equation}
with $r_- / r_+ < 1$. We then have by Hadamard's bound (see, e.g., \cite[Chapter 2.4]{rom}):
\begin{equation}
    \begin{split}
    \left| \sum_{k \geq 0} \frac{(-1)^k}{k!} \sum_{A \subseteq \N, |A|=k} \det [K_c(a_i, a_j)]_{a_i, a_k \in A} \right| &\leq \sum_{k \geq 0} \frac{k^{k/2} C^k}{k!} \sum_{a_1 \geq 0} \cdots \sum_{a_k \geq 0} \prod_{i=1}^k \left(\frac{r_-}{r_+}\right)^{a_i} \\
    &= \sum_{k \geq 0} \frac{k^{k/2} C^k (1-r_-/r_+)^{-k}}{k!} < \infty.
    \end{split}
\end{equation}
This shows $K_c$ is trace class with finite Fredholm determinant and finishes the proof.
\end{proof}

\begin{rem}
    We can also write $\frac{1}{2} D_n^1$ and $D_n^3$ as Fredholm determinants with kernels different from $K_{sp}, K_o$ via the particle--hole involution of e.g.~\cite[Appendix]{boo} since upper bounds on the length of the partition $\ell(\lambda)$ translate into lower bounds on the leftmost \textit{hole} (absence of a particle) in the process $\{ \lambda_i - i + 1/2 \}$ and the latter are Fredholm determinants themselves.
\end{rem}

\begin{rem}
    The conditions on $\rho^{\pm}$ (alternatively on $f$) can probably be relaxed. We refer the reader to~\cite{be4} where this is discussed at length in a functional analytic language. 
\end{rem}

We can express the kernel $K_{sp}(a, b)$ in an alternate form, which we now describe. For that, for a function $g(z)$, define $\tilde{g}(z) := g(1/z)$, and recall the Wiener--Hopf factorization $f = f_+ f_-$. We note $f_{\pm} = H(\rho^{\pm}; z^{\pm 1})$. The function $F$ appearing in $K_{sp}$ is
\begin{equation}
    F = \frac{f_+}{f_- \tilde{f}_-}
\end{equation}
and we denote its Fourier modes by $\left(\frac{f_+}{f_- \tilde{f}_-} \right)_k, k \in \Z$. Let $a' = a+1/2, b' = b+1/2$ to define $K_{sp}$ over $\Z^2$. Using the basic identity
\begin{equation} \label{eq:basic_sp}
    \frac{w}{z} \frac{1}{1-wz^{-1}} + \frac{1}{1-wz} = \frac{(1-w^2)}{(1-wz)(1-wz^{-1})}
\end{equation}
we write---brackets standing for coefficient extraction: 
\begin{equation}
    K_{sp}(a', b') = \left[ \frac{z^{a'}}{w^{b'}} \right] \frac{F(z)}{F(w)} \left( \sum_{i \geq 1} (w z^{-1})^i + \sum_{j \geq 0} (wz)^j \right).
\end{equation}
After expanding $F(z), 1/F(w)$ in Fourier series, we obtain
\begin{equation}
    K_{sp}(a', b') = \left( \frac{f_+}{f_- \tilde{f}_-} \right)_{a'} \left( \frac{f_- \tilde{f}_-}{f_+} \right)_{-b'} + \sum_{j \geq 1} \left( \frac{f_- \tilde{f}_-}{f_+} \right)_{-b'-j} \left( \left( \frac{f_+}{f_- \tilde{f}_-} \right)_{a'+j} + \left( \frac{f_+}{f_- \tilde{f}_-} \right)_{a'-j} \right).
\end{equation}

A similar manipulation can be done for $K_o$. To wit, if $a''=a-1/2, b''=b-1/2$ and owing to the identity
\begin{equation} \label{eq:basic_o}
    \frac{z}{w} \left( \frac{1}{1-w z^{-1}} - \frac{1}{1-wz} \right) = \frac{1-z^2}{(1-w z)(1-w z^{-1})},
\end{equation}
one obtains after similar computations
\begin{equation}
    K_{o}(a'', b'') = \sum_{j \geq 0} \left( \frac{f_- \tilde{f}_-}{f_+} \right)_{-b''-j} \left( \left( \frac{f_+}{f_- \tilde{f}_-} \right)_{a''+j} - \left( \frac{f_+}{f_- \tilde{f}_-} \right)_{a''-j} \right).
\end{equation}

\section{Asymptotics of Plancherel-like measures} \label{sec:asymptotics}

In this section we asymptotically analyze certain symplectic and orthogonal analogues of the Poissonized Plancherel measure on partitions. We do this mostly to show, on simple if contrite examples, old and new behavior that one \textit{generically} expects asymptotically. We consider both bulk and edge scaling. 

Fix a parameter $\theta > 0$. Let $pl_\theta$ be the Plancherel specialization of $\Lambda$ defined on the Newton powersums by $p_1(pl_\theta) = \theta$ and $p_k(pl_\theta) = 0, k \geq 2$. It is alternatively defined by $h_n(pl_\theta) = \theta^n/n!$ for $n\geq 0$. Consider the \textit{signed} measures $P_{sp/o}$ defined by
\begin{equation}
    P_{sp} (\lambda) = \exp\left(-\frac{3 \theta^2}{2}\right) \cdot sp_{\lambda} (pl_{2 \theta}) s_{\lambda} (pl_{\theta}), \qquad P_{o} (\lambda) = \exp\left(-\frac{3 \theta^2}{2}\right) \cdot  o_{\lambda} (pl_{2 \theta}) s_{\lambda} (pl_{\theta}).
\end{equation}
These measures are symplectic and orthogonal analogues of the Poissonized Plancherel measure\footnote{The choice of $2\theta$ inside $sp/o$ is to facilitate the computations, but could be replaced by any other parameter different from $0, \theta$ if one is willing to do a change of variables at the end. The empty specialization $0$ for $sp/o$ is also interesting, and related to Section 7 of~\cite{rai}, but we defer this to another time and space.}. For fixed $\theta$ they are not necessarily positive. However as $|\lambda|$ increases $P_{sp/o}(\lambda)$ does become positive since $sp/o_{\lambda}(pl_{2 \theta})$ is a polynomial in $\theta$ with leading monomial $\theta^{|\lambda|}$ having positive coefficient due to the Jacobi--Trudi formulas~\eqref{eq:sp_o_jt} for $sp$ and $o$.

An alternative description of these measures can be given as follows. Recall that for $\mu \subseteq \lambda$, the skew Schur function $s_{\lambda/\mu}$, evaluated at Plancherel, is 
\begin{equation}
    s_{\lambda / \mu} (pl_{\theta}) = \theta^{|\lambda/\mu|} \frac{\dim (\lambda / \mu)}{|\lambda/\mu|!}
\end{equation}
where $\dim (\lambda/\mu)$ is the number of standard Young tableaux of shape $\lambda/\mu$, and $|\lambda/\mu| = |\lambda| - |\mu|$. With $\mu=0$, $\dim \lambda$ is the dimension of the irreducible representation---indexed by $\lambda$---of the symmetric group $S_{|\lambda|}$. Generically, $\dim (\lambda/\mu)$ is also a dimension, that of the (reducible) representation of $S_{|\lambda/\mu|}$ on $\hom_{S_{|\mu|}} (\mu, \lambda)$ where we view $\mu$ and $\lambda$ as the corresponding irreducibles of $S_{|\mu|}$ and $S_{|\lambda|}$---see~\cite[Section 3.5]{cst} for more details. Finally, using~\eqref{eq:sp_o_s} with $pl_{2 \theta}$ replacing $X, \tilde{X}$, we see that
\begin{equation}
    \begin{split}
    P_{sp}(\lambda) \ &\propto \ \theta^{|\lambda|} \frac{\dim \lambda}{|\lambda|!} \sum_{\alpha} (-1)^{|\alpha|/2} (2 \theta)^{|\lambda/\alpha|} \frac{\dim (\lambda / \alpha)}{|\lambda / \alpha|!}, \\
    P_{o}(\lambda) \ &\propto \ \theta^{|\lambda|} \frac{\dim \lambda}{|\lambda|!} \sum_{\beta} (-1)^{|\beta|/2} (2 \theta)^{|\lambda/\beta|} \frac{\dim (\lambda / \beta)}{|\lambda / \beta|!}
    \end{split}
\end{equation}
where the sums range over partitions $\alpha, \beta$ having Frobenius coordinates $(a_1, a_2, \dots|a_1+1, a_2+1, \dots)$ and $(b_1+1, b_2+1, \dots | b_1, b_2, \dots)$ respectively.

Theorem~\ref{thm:corr_lifted} states the measures $P_{sp/o}$ are determinantal with kernels $K_{sp/o}$ governed by the same function $F$:
\begin{equation}
    F(z) = \exp(\theta(z - z^{-1})) = \sum_{n \in \Z} z^n J_n (2 \theta)
\end{equation}
where $J_n$ is a Bessel function of the first kind~\cite{wat}. We also note here that the function $f$ from Section~\ref{sec:gboth} is
\begin{equation}
    f(z) = \exp(\theta(2z+z^{-1})) = \exp(\theta \sqrt{2} (z \sqrt{2} + (z \sqrt{2})^{-1})) = \sum_{n \in \Z} (z \sqrt{2})^n I_n (2 \sqrt{2} \theta)
\end{equation}
where $I_n = \im^{-n} J_n$ is a Bessel function of the second kind.

Throughout this section we use the convention that the kernels $K_{sp/o}$ are defined over $\Z\times\Z$ \footnote{As we have left Fock space long ago.} and further assume that every quantity in need of being an integer is such without surrounding it with floor brackets. 

We start by examining the bulk behavior of both measures. As the argument is almost identical to that of~\cite[Theorem 3]{oko3} or of~\cite[Theorem 2 and Corollary 3]{or}, we only sketch it. As $\theta \to \infty$, we scale the entries $a, b$ in the double contour integral representation of the kernel $K_{sp/o}(a, b)$ linearly:
\begin{equation} \label{eq:bulk_scaling}  
    a \approx \alpha \theta + \a, \qquad b \approx \alpha \theta + \b, \qquad \a, \b \in \Z, \ \alpha \in \R \text{ fixed}.
\end{equation}

We use the relations~\eqref{eq:basic_sp} and~\eqref{eq:basic_o} to split each kernel into a sum of two integrals. The exponentially dominant term in the asymptotics is $\exp ( \theta ( S(z)-S(w) ) )$ for the action
\begin{equation}\label{eq:S_action}
    S(z):= z - z^{-1} - \alpha \log z.
\end{equation}
The function $S$ has two critical points $z_\pm = 2^{-1} (\alpha \pm \sqrt{\alpha^2-4})$ which are real if $|\alpha| > 2$ and complex conjugate on the unit circle if $|\alpha| < 2$---the case $|\alpha|=2$ corresponds to edge behavior and will be treated separately. 

Suppose we are in the complex case $|\alpha| < 2$, and let $\phi_+ := \arg z_+ \in (0, \pi)$. On the unit circle $|z| = 1$, the function $\Re S(z) - \Re S(z_+)$ is zero, with increasing value towards the interior of the unit circle if we're on the left of the vertical line connecting $z_\pm$, and increasing towards the exterior if we're on the right of said vertical line. Recall that in $K_{sp/o}$, the $z$ contour is on the outside of the $w$ contour, and we can take without loss of generality the $z$ contour outside of the unit circle and the $w$ contour on the inside. We then pass both the $z$ contour and the $w$ contour through the two critical points $z_\pm$ in the direction of descent for $z$ and ascent for $w$ so that in the end, on the right of the vertical line connecting $z_\pm$, the $w$ contour is outside the unit circle and the $z$ contour inside---on the left of said vertical line nothing changes. The new integrals will converge exponentially fast to zero but in exchanging the contours we will have picked up some residues. Due to equations~\eqref{eq:basic_sp} and~\eqref{eq:basic_o}, there are two types of residues. The first is the classical one of~\cite[Theorem 3]{oko3} or of~\cite[Theorem 2 and Corollary 3]{or} at $w z^{-1} = 1$ or $z=w$---corresponding to the first summand from~\eqref{eq:basic_sp} and~\eqref{eq:basic_o}---and the contribution here is a single contour integral over an arc connecting $z_+$ to $z_-$, oriented bottom to top and passing to the right of $0$---let us call this contour $\gamma$. This is the \textit{discrete sine} part, the integral in question converging to 
\begin{equation} \label{eq:sine_kernel}
    \frac{1}{2 \pi \im} \int_{\gamma} w^{\b - \a - 1} \dx w = \frac{\sin (\phi_+ (\b-\a))}{\pi (\b-\a)} := K_{sine}(\a, \b).
\end{equation}
The second residue is at $zw=1$. If we're doing the deforming with $w$ first then we will hit all the poles at $w=1/z$ on a contour $\gamma'$ connecting $z_\pm$ from bottom to top which is \textit{inside} the unit circle and to the right of the origin. The corresponding residue
\begin{equation}
    \frac{1}{2 \pi \im} \int_{\gamma'} \exp(2 \theta (z - z^{-1} - \alpha \log z)) z^{-\a -\b - 1} \dx z
\end{equation}
will still converge exponentially fast to zero due to $\gamma'$ being in the area of descent of $\Re S$---inside the unit circle and to the right of the vertical line connecting $z_\pm$. All in all if $|\alpha| < 2$, $K_{sp/o}(a, b)$ converges to the discrete sine kernel $K_{sine}$ of~\eqref{eq:sine_kernel}.  

For $\alpha > 2$, $z_\pm$ are real, $z_+ > 1, 0 < z_- < 1$. We can then just enlarge the $z$ contour to pass through $z_+$ and shrink the $w$ contour to pass through $z_-$ without ever exchanging the contours (or a contour with its reciprocal). A simple inspection of $\Re S$ shows that the resulting integral will go to $0$. For $\alpha < -2$ we have $z_- < -1, z_+ \in (-1, 0)$ and if we enlarge the $w$ contour (respectively shrink the $z$ contour) so that the $w$ passes through $z_-$ and $z$ through $z_+$, the resulting double contour integral will again go to $0$ and we pick up the possibly non-zero usual $z=w$ residue---exactly like in the case $|\alpha| < 2$ but now on a whole circle around the origin. This residue is now $0$ unless $\b=\a$ in which case it is $1$. Both cases are in accordance to and can be seen as limits of~\eqref{eq:sine_kernel} for $\phi_+ = \pi$ if $\alpha < -2$ and $\phi_+ = 0$ if $\alpha > 2$.   

We thus arrive at the following.
\begin{thm} \label{thm:bulk_scaling}
    Let $k_i = \alpha \theta + \a_i$ as $\theta \to \infty$ for fixed $\a_i \in \Z, \alpha \in \R$. Then bulk correlations converge to determinants of the discrete sine kernel:
    \begin{equation}   
        P_{sp/o} ( \{ k_1, \dots, k_n \} \subset \{ \lambda_i - i \}) \to \det \left[ \frac{\sin (\phi_+(\a_i - \a_j)) }{\pi (\a_i - \a_j)} \right]_{1 \leq i, j \leq n}  
    \end{equation}
    where $\phi_{+} = \arg 2^{-1} (\alpha + \sqrt{\alpha^2-4})$.
\end{thm}

We now examine the edge, when the action $S$ has a double critical point at $z=1$ due to $\alpha \in \pm 2$. For clarity we take $\alpha = 2$, the case $\alpha=-2$ following by symmetry.

Define the following kernels:
\begin{equation}
    \begin{split}
    \Apm(x, y) :&= \int_{0}^{\infty} Ai(x+s) Ai(y+s) \dx s \pm \int_{0}^{\infty} Ai(x-s) Ai(y+s) \dx s \\
        &= \frac{1}{(2 \pi \im)^2} \left( \int_{\gamma_R} \dx u \int_{\gamma_L} \dx v \frac{\exp(u^3/3 - y u)}{\exp(v^3/3 - x v)} \frac{1}{u-v} \pm \int_{\gamma_R} \dx u \int_{\gamma_L} \dx v \frac{\exp(u^3/3 - y u)}{\exp(v^3/3 - x v)} \frac{1}{u+v} \right) \\
        &= \frac{1}{(2 \pi \im)^2} \left( \int_{\gamma'_L} \dx \omega \int_{\gamma'_R} \dx \zeta \frac{\exp(\zeta^3/3 - x \zeta)}{\exp(\omega^3/3 - y \omega)} \frac{1}{\zeta-\omega} \pm \int_{\gamma'_L} \dx \omega \int_{\gamma'_R} \dx \zeta \frac{\exp(\zeta^3/3 - x \zeta)}{\exp(\omega^3/3 - y \omega)} \frac{1}{-(\zeta+\omega)}  \right)
    \end{split}
\end{equation}
with the contours in the second equation defined as follows: $\gamma_R$ goes from somewhere between $\exp(-\pi \im /3) \cdot \infty$ and $-\im \cdot \infty$ at the bottom to somewhere between $\exp(\pi \im /3) \cdot \infty$ and $\im \cdot \infty$ at the top, while staying to the right of $\im \R$; $\gamma_L$ goes from somewhere between $\exp(-2 \pi \im /3) \cdot \infty$ and $-\im \cdot \infty$ at the bottom to somewhere between $\exp(2 \pi \im /3) \cdot \infty$ and $\im \cdot \infty$ at the top, while staying to the left of $\im \R$; $\gamma_L$ stays closer to the imaginary axis $\im \R$ than $-\gamma_R$ so that the latter is completely to the left of the former; and finally $\gamma'_L = -\gamma_R, \gamma'_R = - \gamma_L$. The contours $\gamma_R, \gamma_L$ are those of Borodin--Ferrari--Sasamoto~\cite{bfs}, but we find the primed ones more convenient in the sequel. The Airy function $Ai(x)$ is defined by
\begin{equation}
    Ai(x) = \int_C \exp\left( \frac{\zeta^3}{3} - x \zeta \right) \frac{\dx \zeta} {2 \pi \im}
\end{equation}
over the contour $C$ which goes in a straight line from $\exp(-\pi \im /3) \cdot \infty$ to 0 and then from $0$ to $\exp(\pi \im /3) \cdot \infty$. Note we can deform $C$ into any contour of the aforementioned form $\gamma_R$ without changing the value of the integral. Finally the equivalence between the first two lines is the simple evaluation $(u \pm v)^{-1} = \int_0^\infty \exp(-s(u \pm v))$ for $\Re(u \pm v) > 0$\footnote{It is because of this that $\gamma_L$ stays to the right of $-\gamma_R$.}, while for the last line $(\omega, \zeta) = -(u, v)$. $\Ap$, as defined by the first two lines, is the usual \textit{Airy 2 to 1} kernel of~\cite[Equation 2.7]{bfs}, while $\Am$ is another such variant. 

\begin{thm} \label{thm:edge_scaling}
    Fix $x, y \in \R$ and let 
\begin{equation}
        a \approx 2 \theta + x \theta^{1/3}, \qquad b \approx 2 \theta + y \theta^{1/3}, \qquad \text{as\ } \theta \to \infty.
\end{equation}
Then
\begin{equation}
    \theta^{1/3} K_{sp}(a, b) \to \Ap(x, y), \qquad \theta^{1/3} K_{o}(a, b) \to \Am(x, y)
\end{equation}
uniformly for $x, y$ ranging over compact sets. 
\end{thm}

\begin{proof}
We give two proofs of this result, concentrating on the symplectic case and kernel. The orthogonal case follows mutatis-mutandis.

The first proof, in which we sketch the main argument, is similar to the proof of Proposition 4.1 in~\cite{boo}. For a more detailed exposition and much more, see the~\cite[Chapter 2]{rom}---especially Theorem 2.27. A computation similar to that at the end of Section~\ref{sec:gboth} shows
\begin{equation} \label{eq:kernels_bessel}
    \begin{split}
    K_{sp} (a, b) &= \sum_{i \geq 1} J_{a+i} (2 \theta) J_{b+i} (2 \theta) + \sum_{i \geq 0} J_{a-i} (2 \theta) J_{b+i} (2 \theta),\\
    K_{o} (a, b) &= \sum_{i \geq 0} J_{a+i} (2 \theta) J_{b+i} (2 \theta) - \sum_{i \geq 0} J_{a-i} (2 \theta) J_{b+i} (2 \theta).
    \end{split}
\end{equation}

We use Nicholson's approximation---found in~\cite[Section 8.43]{wat} but see~\cite[Lemma 4.4]{boo} for a recent elementary proof. We state the slightly stronger form of Johansson~\cite[Lemma 3.9]{joh4}. Given $M > 0$, there exists a constant $C > 0$ such that
\begin{align}
    |\theta^{1/3} J_{2 \theta + x \theta^{1/3}} (2 \theta) | &\leq C \exp \left(-\frac{|x|}{4} \min\left(\frac{1}{4} \theta^{1/3}, \sqrt{|x|}\right) \right), \qquad \forall x > -M; \label{nic_decay} \\
    \lim_{\theta \to \infty} \theta^{1/3} J_{2 \theta + x \theta^{1/3}} (2 \theta) &= Ai(x), \qquad \text{uniformly for\ } x \in [-M, M]. \label{eq:nic_limit}
\end{align}
Armed with this, we see that as $\theta \to \infty$ the four sums of~\eqref{eq:kernels_bessel}
\begin{equation}
    \theta^{1/3} \sum_{i} J_{a \pm i} (2 \theta) J_{b + i} (2 \theta) \to \int_0^\infty Ai(x \pm s) Ai(y + s) \dx s
\end{equation}
following the same steps as in the proof of~\cite[Proposition 4.4]{boo} or of~\cite[Theorem 1.4]{joh4}\footnote{See also the pedagogical proof of Theorem 2.7(a) from~\cite{rom}.}. The left hand side above becomes a Riemann sum for the right hand side to which it converges as $\theta \to \infty$.

The second proof uses the method of steepest descent, much like the proof for the bulk. As again the arguments are well established by now, we sketch them. Recall that in the integral representation of $K_{sp}$ we can take the contours to be $|z| = 1 + \delta, |w| = 1 - \delta$ for some small $\delta>0$. The ratio $z^{-a} w^b F(z)/F(w)$, in the regime of interest, becomes
\begin{equation} 
    \frac{z^{-a}F(z)}{w^{-b}F(w)} \approx \frac{\exp(\theta (z - z^{-1}-2 \log z)-x \theta^{1/3} \log z )}{\exp(\theta (w - w^{-1}-2 \log w)-y \theta^{1/3} \log w )}
\end{equation}
Using the method of steepest descent with action $S$ from~\eqref{eq:S_action}, we pass the contours through the (technically close to the) double critical point $z_0=1$ such that locally outside 1, the $z$ contour follows the direction of steepest descent---makes angles $\pm \pi/3$ with the horizontal axis and is outside the unit circle, while the $w$ contour follows the direction of steepest ascent---angles $\pm 2 \pi /3$ and inside the unit circle. The main asymptotic contribution comes from a $\theta^{1/3}$ neighborhood of the double critical point $z_0 = 1$. Outside this neighborhood the integral vanishes exponentially if the $z$ contour stays outside the unit circle (so the real part of $S$ is negative) and the $w$ one inside---here and above we note $S'''(z) = 2 > 0$ forces steepest descent for $z$ and steepest ascent for $w$. We thus scale the $z, w$ variables as 
\begin{equation}
    z = \exp(\zeta \theta^{-1/3}), \qquad w = \exp(\omega \theta^{-1/3}). 
\end{equation}
In this scaling we have
\begin{equation}
    \frac{1-w^2}{(1-wz^{-1})(1-wz)} \to \frac{2 \omega}{(\zeta-\omega)(\zeta+\omega)} = \frac{1}{\zeta-\omega} + \frac{1}{-(\zeta+\omega)}
\end{equation}
while the integrand becomes 
\begin{equation}
    \frac{\exp(\zeta^3/3 - x \zeta)}{\exp(\omega^3/3 - y \omega)}
\end{equation}
with $\dx z \dx w / (zw) = \theta^{-1/3} \dx \zeta \dx \omega$ (it is because of this that we need to multiply the discrete kernel by $\theta^{1/3}$). The contours have now the correct local behavior around 0 in the new coordinates $\zeta, \omega$ and we can continue them to $\infty$ in the stated directions of $\gamma'_R$ for $\zeta$ and $\gamma'_L$ for $\omega$ at the cost of exponentially negligible terms. The reader interested in more precise estimates and details can consult, e.g.,~\cite{bfs}. Putting it all together, we arrive at the stated (third line) expression for the limit kernel. 
\end{proof}

\begin{rem}
    The kernels $K_{sp}$ and $K_o$ are related to one another by a duality---Macdonald's involution at the level of proofs. There should be a similar \textit{continuous} duality between
    \begin{equation}  
        \Ap(x, y) = \frac{1}{(2 \pi \im)^2} \int_{\gamma_R} \dx u \int_{\gamma_L} \dx v \frac{\exp(u^3/3 - y u)}{\exp(v^3/3 - x v)} \frac{2 u}{(u-v)(u+v)} 
    \end{equation}
    and 
    \begin{equation}
        \Am(x, y) = \frac{1}{(2 \pi \im)^2} \int_{\gamma_R} \dx u \int_{\gamma_L} \dx v \frac{\exp(u^3/3 - y u)}{\exp(v^3/3 - x v)} \frac{2 v}{(u-v)(u+v)} 
    \end{equation}
    but at the moment we cannot make this precise.
\end{rem}

\begin{rem}
    The same analysis for the original \textit{Poissonized Plancherel} measure $PP(\lambda) \propto s_{\lambda}^2 (pl_{\theta})$ yields the Airy$_2$ kernel $\mathcal{A}_2(x, y) = \int_{0}^{\infty} Ai(x+s) Ai(y+s) \dx s$ in the edge limit. From this perspective $\Apm$ should be viewed as \textit{type B and C} analogues of $\mathcal{A}_2$. This also suggests $\Apm$ should be universal for edge behavior of symplectic and orthogonal measures---regardless of specializations, as long as the limiting action $S$ has a double critical point at the edge, much like $\mathcal{A}_2$ is universal for Schur measures with such properties. Of course this is simply the consequence of the factors $(1-w^2) (1-wz^{-1})^{-1} (1-wz)^{-1}$ and $(1-z^2) (1-wz^{-1})^{-1} (1-wz)^{-1}$ present in the integrals for $K_{sp/o}$, replacing the usual factor $(1-wz^{-1})$ present in the case of Schur measures~\cite{oko}.
 \end{rem} 

As the finite discrete kernels are trace class and the $\Apm$ kernels are conjugate to such, a computation similar to that of the proof of Theorem~\ref{thm:edge_scaling} (see, e.g.,~\cite{boo}) shows that if $a = 2 \theta + \theta^{1/3} s$ then 
\begin{equation}
    (\tr K)_{\ell^2 \{a, a+1, \dots \}} = \sum_{b \geq a} K_{sp/o}(b, b) \to \text{``} (\tr \Apm)_{L^2 (s, \infty)} = \int_s^\infty \Apm(x,x) \dx x \text{''}
\end{equation}
where we take $+$ for $sp$ and $-$ for $o$. That is, Fredholm determinants converge to Fredholm determinants.

\begin{rem}
    The Fredholm determinant $\det(1-\Ap)_{L^2(s, \infty)}$ is the Tracy--Widom \textit{2 to 1} distribution of Borodin, Ferrari and Sasamoto~\cite{bfs}. It exists because the kernel $\Ap$ can be conjugated to a trace class kernel~\cite[Appendix B]{bfs}\footnote{A similar argument shows the existence of $\det(1-\Am)_{L^2(s, \infty)}$.} and cf. op. cit. it appears in the scaling of TASEP with half-flat half-wedge initial conditions at the transition between the two halves. It also appears in the scaling of the point-to-half-line last passage percolation with exponential weights, as per recent work of Bisi--Zygouras~\cite[Theorem 2.1]{bz2} and~\cite{bz1}. It interpolates between the Tracy--Widom GUE $F_2$ and GOE $F_1$ distributions. 
\end{rem} 

We then have the following.

\begin{thm} \label{thm:edge_scaling_2}
    For $n = 1,2,3,4\dots$ fixed, the largest first $n$ points of $\{\lambda_i - i \}$ in the edge scaling of Theorem~\ref{thm:edge_scaling} converge, in the sense of finite dimensional distributions, to the first $n$ points of the $\Ap$ process for $P_{sp}$ and the $\Am$ process for $P_o$. In particular, we have:
    \begin{equation}
    \lim_{\theta \to \infty} P_{sp}\left(\frac{\lambda_1 - 2 \theta}{\theta^{1/3}} \leq s\right) = \det(1-\Ap)_{L^2(s, \infty)}, \qquad \lim_{\theta \to \infty} P_{o}\left(\frac{\lambda_1 - 2 \theta}{\theta^{1/3}} \leq s\right) = \det(1-\Am)_{L^2(s, \infty)}.
    \end{equation}
\end{thm}

\section{Conclusion and future outlook} \label{sec:conclusion}

In this section we conclude by providing an outlook for where the techniques and results herein announced might be useful.

First, on the Toeplitz+Hankel determinant side, both in~\cite{br1} and in~\cite{be4} there are other such determinants not considered here. In particular the latter discusses constants (in their notation) $E F_{I}, E F_{II}$ which are simple modifications of $Z_{sp}, Z_{o}$ from above. One can expand these new constants easily in Schur functions, call the coefficients of $s_{\lambda}$ e.g. $sq_{\lambda}, sr_{\lambda}$ and then \textit{provided one finds an appropriate Jacobi--Trudi identity}, the other results from~\cite{be4} should follow. 

Second, the symplectic Schur measure afore described, specialized in variables, is related, at least in the continuous (exponential) limit, to last passage percolation in the so-called point-to-half-line regime, as one can see in~\cite{bz1, bz2}. Along with the cited authors, we hope to clarify and address this in a future note.

Third, it would be interesting to clarify the so-called integrable hierarchy behind the symplectic and orthogonal Schur functions. Note Baker~\cite{bak} already discusses this in a limited yet informative way. We do not know in particular if any interesting (to other branches of mathematics) $\tau$-functions come out of Baker's construction.  

Fourth, the symplectic and orthogonal measures from above, specialized with $\rho^+=0$, yield Schur measures described by Rains~\cite[Section 7]{rai}. Our methods already recover Theorems 7.1 and 7.2 from op. cit. Furthermore, simple back-of-the-envelope computations show the asymptotic behavior of such measures (and the associated longest increasing subsequences) is governed by a bulk \textit{sine$\pm$sine} kernel and an edge $\Apm$ kernel. We address this in a future note.

Fifth, one can try to generalize our results to the case of a symplectic/orthogonal equivalent of the Schur \textit{process}~\cite{or}. While the dynamical correlations can be written down naively with minimal effort, the lack of a symplectic/orthogonal branching rule compatible with the corresponding half-vertex operators prevented the author from writing down the result. It is nonetheless possible some traction could be gained here.

Sixth, the results herein obtained along with the work of Bisi--Zygouras~\cite{bz1, bz2} seem to suggest certain otherwise thought to be pfaffian processes also have determinantal representations. Investigating whether this extends to e.g. \textit{all} pfaffian processes of~\cite{rai} would be worthwhile for a simple counting reason: when dealing with pfaffians, one usually has a $2\times 2$ matrix kernel the analysis of which requires at least three different cases (for the diagonal and the $2,1$ entries). Determinantal processes described so far only have one simple kernel. $1 < 3$.

Seventh, a combinatorial proof (of a disguised form) of the Szeg\H{o} Theorem~\ref{thm:szego} above can be found in~\cite{deh}---see also references therein. It uses Haar averages over the orthogonal and symplectic groups whereas we use the Cauchy identity. Presumably a common framework for both would be that of Gelfand pairs and vanishing integrals. We cannot further comment on this in any meaningful way at the moment, but it is a point that deserves clarification.

\bibliographystyle{myhalpha}
\bibliography{symplectique}

\end{document}